\documentclass[10pt]{amsart}

\usepackage{amssymb,bm}
\usepackage[all]{xy}

\usepackage{etoolbox}

\DeclareMathAlphabet{\mathbbm}{U}{bbm}{m}{n}

\DeclareMathAlphabet{\bbi}{U}{bbm}{m}{sl}

\usepackage[
	pdfborder={0 0 0},
	colorlinks,
	plainpages,
	backref,
	citecolor=red!50!black,
	filecolor=Darkgreen,
	linkcolor=blue!40!black,
	urlcolor=cyan!50!black!90
]{hyperref}

\usepackage{graphicx,accents}

\makeatletter
\DeclareRobustCommand{\cev}[1]{%
  \mathpalette\do@cev{#1}%
}
\newcommand{\do@cev}[2]{%
  \fix@cev{#1}{+}%
  \reflectbox{$\m@th#1\vec{\reflectbox{$\fix@cev{#1}{-}\m@th#1#2\fix@cev{#1}{+}$}}$}%
  \fix@cev{#1}{-}%
}
\newcommand{\fix@cev}[2]{%
  \ifx#1\displaystyle
    \mkern#23mu
  \else
    \ifx#1\textstyle
      \mkern#23mu
    \else
      \ifx#1\scriptstyle
        \mkern#22mu
      \else
        \mkern#22mu
      \fi
    \fi
  \fi
}

\makeatletter
\DeclareRobustCommand{\ccong}{\mathrel{\mathpalette\@verequiv\sim}}
\newcommand{\@verequiv}[2]{%
  \lower.5\p@\vbox{
    \lineskiplimit\maxdimen
    \lineskip-.5\p@
    \ialign{%
      $\m@th#1\hfil##\hfil$\crcr
      #2\crcr
      \equiv\crcr
    }%
  }%
}
\makeatother

\numberwithin{equation}{section}

\newtheorem{thrm}{Theorem}[section]

\newtheorem{lemma}[thrm]{Lemma}

\newtheorem{defn}[thrm]{Definition}

\theoremstyle{remark}
\newtheorem{rmk}[thrm]{Remark}
\newtheorem{exam}[thrm]{Example}

\newcommand{\nc}{\newcommand}

\nc{\ta}{\tilde{a}}

\nc{\al}{\alpha}
\nc{\be}{\beta}
\nc{\eps}{\epsilon}
\nc{\veps}{\varepsilon}
\nc{\ga}{\gamma}
\nc{\Ga}{\Gamma}
\nc{\ka}{\kappa}
\nc{\vka}{\varkappa}
\nc{\la}{\lambda}
\nc{\La}{\Lambda}
\nc{\del}{\delta}
\nc{\om}{\omega}
\nc{\si}{\sigma}
\nc{\vsi}{\varsigma}
\nc{\Ups}{\upsilon}
\nc{\vphi}{\varphi}

\nc{\ud}{\underline}
\nc{\tl}{\tilde}

\nc{\at}{\tilde{a}}
\nc{\tu}{\tilde{u}}
\nc{\tv}{\tilde{v}}
\nc{\tw}{\tilde{w}}

\nc{\mfg}{\mathfrak{g}}
\nc{\mfh}{\mathfrak{h}}
\nc{\mfso}{\mathfrak{so}}
\nc{\mfsp}{\mathfrak{sp}}
\nc{\mfgl}{\mathfrak{gl}}

\nc{\mfL}{\mathfrak{L}}
\nc{\mcL}{\mathcal{L}}

\nc{\End}{\mathrm{End}}
\nc{\Ext}{\mathrm{Ext}}
\nc{\Hom}{\mathrm{Hom}}
\nc{\Ima}{\mathrm{Image}}
\nc{\Ind}{\mathrm{Ind}}
\nc{\Ker}{\mathrm{Ker}}
\nc{\RHom}{\mathrm{RHom}}
\nc{\Sym}{\mathrm{Sym}}

\nc{\ddeg}{\mathtt{deg}}
\nc{\dimm}{\mathtt{dim}}
\nc{\lmod}{\mathtt{lmod}}
\nc{\opp}{\mathtt{opp}}
\nc{\rmod}{\mathtt{rmod}}

\nc{\mmod}{\mathrm{mod}}
\nc{\nbh}{\mathrm{nbh}}

\nc{\mf}{\mathfrak}
\nc{\mc}{\mathcal}
\nc{\ms}{\mathsf}
\nc{\bb}{\mathbb}
\nc{\mr}{\mathscr}

\nc{\wh}{\widehat}
\nc{\wt}{\widetilde}

\nc{\R}{\mathbb{R}}
\nc{\Q}{\mathbb{Q}}
\nc{\C}{\mathbb{C}}
\nc{\N}{\mathbb{N}}
\nc{\Z}{\mathbb{Z}}

\nc{\ot}{\otimes}
\nc{\op}{\oplus}
\nc{\ol}{\overline}

\DeclareMathOperator{\sign}{sign}
\DeclareMathOperator{\spanc}{span_\C}

\nc{\bdot}{\boldsymbol\cdot}

\nc{\lp}[1]{\ell^+_{#1}}
\nc{\lm}[1]{\ell^-_{#1}}
\nc{\lpm}[1]{\ell^\pm_{#1}}
\nc{\lmp}[1]{\ell^\mp_{#1}}

\nc{\Res}[1]{\underset{\;{#1}\;}{\rm Res}}

\nc{\equ}[1]{\begin{equation}#1\end{equation}}
\nc{\eqa}[1]{\begin{equation}\begin{alignedat}{50}#1\end{alignedat}\end{equation}}
\nc{\eqn}[1]{\begin{equation*}\begin{alignedat}{50}#1\end{alignedat}\end{equation*}}
\nc{\eqg}[1]{\begin{equation}\begin{gathered}#1\end{gathered}\end{equation}}

\nc{\ali}[1]{\begin{alignat}{50}#1\end{alignat}}
\nc{\als}[1]{\begin{subequations}\begin{alignat}{50}#1\end{alignat}\end{subequations}}
\nc{\aln}[1]{\begin{alignat*}{50}#1\end{alignat*}}

\nc{\gat}[1]{\begin{gather}#1\end{gather}}
\nc{\gas}[1]{\begin{subequations}\begin{gather}#1\end{gather}\end{subequations}}
\nc{\gan}[1]{\begin{gather*}#1\end{gather*}}

\usepackage{color}  
\usepackage[usenames,dvipsnames]{xcolor}
\nc{\red}{\color{red}}
\nc{\blu}{\color{blue}}
\nc{\gre}{\color{green}}
\nc{\orn}{\color{orange}}
\nc{\brn}{\color{Brown}}
\nc\el{\nonumber\\}
\nc\nn{\nonumber}

\DeclareMathOperator{\tr}{tr}

\nc{\qu}{\quad}
\nc{\qq}{\qquad}

\renewcommand{\,}{\kern 0.1em} % standard is \kern 0.16667em

\nc{\bi}{\bar\imath}
\nc{\bj}{\bar\jmath}
\nc{\bk}{\bar k}
\nc{\bl}{\bar l}

\usepackage[scr=boondoxo]{mathalfa}
\nc{\key}{{\mathscr{k}}}
\nc{\ley}{{\mathscr{l}}}
\nc{\hey}{{\mathscr{h}}}
\nc{\emm}{{\mathscr{m}}}

\begin{document}

\title[Nested algebraic Bethe ansatz for deformed orthogonal and symplectic spin chains]{Nested algebraic Bethe ansatz for deformed \\ orthogonal and symplectic spin chains}

\begin{abstract} 
We construct exact eigenvectors and eigenvalues for $U_q(\mfsp_{2n})$- and $U_q(\mfso_{2n})$-symmetric closed spin chains by means of a nested algebraic Bethe ansatz method. We use a fusion procedure to construct higher-dimensional Lax operators. Our approach generalises and extends the results obtained by Reshetikhin and De Vega--Karowski. We also present a generalisation of Tarasov--Varchenko trace formula for nested Bethe vectors.
\end{abstract}

\author{Allan Gerrard}
\address{University of York, Department of Mathematics, York, YO10 5DD, UK.}
\email{ajg569@york.ac.uk}

\author{Vidas Regelskis}
\address{University of Hertfordshire, School of Physics, Astronomy and Mathematics, Hatfield AL10 9AB,
UK, and
Vilnius University, Instiute of Theoretical Physics and Astronomy, Saul\.etekio av.~3, Vilnius 10257, Lithuania.}
\email{vidas.regelskis@gmail.com}

\maketitle

\setlength{\parskip}{1ex}

%%%%%%%%%%%%%%%%%%%%%%%%%%%%%%%%%%%%%%%%%%%%%%%%%%%%%%%%%%%%%%%%%%
% 
%%%%%%%%%%%%%%%%%%%%%%%%%%%%%%%%%%%%%%%%%%%%%%%%%%%%%%%%%%%%%%%%%%

\section{Introduction}

The nested algebraic Bethe ansatz is a large collection of methods used to find eigenvectors and eigenvalues of transfer matrices of lattice integrable models associated with higher rank Lie algebras. It effectively reduces the problem of diagonalizing quantum Hamiltonians to a set of algebraic equations, known as the Bethe ansatz equations, that in many cases can be solved using numerical methods, see reviews \cite{PRS18,Sla07}. In addition, the nested Bethe vectors play an important role not only in the theory of quantum integrable models but also in representation theory of quantum groups more generally. For instance, Bethe vectors can be used to construct Jackson integral representations for solutions of the quantized Knizhnik--Zamolodchikov equations \cite{Res92,TaVa94}. 

The nested algebraic Bethe ansatz for orthogonal and symplectic spin chains in the rational setting was studied by De Vega and Karowski \cite{DVK87} and Reshetikhin \cite{Res85,Res91}. The analytic Bethe ansatz in both rational and trigonometric setting has been addressed by Kuniba and Suzuki in \cite{KuSu95}. An extension towards models with rational orthosymplectic symmetries via a different nesting procedure was presented by Martins and Ramos in \cite{MaRa97}. However there has been little progress in orthogonal and symplectic integrable models with periodic boundary conditions since. The $O(N)$ Gross-Neveu model has been recently studied by Babujian, Foerster and Karowski in  \cite{BMF12, BMF16}.

This paper aims to fill gaps in the literature by presenting a nested algebraic Bethe ansatz for $U_q(\mfsp_{2n})$- and $U_q(\mfso_{2n})$-symmetric closed spin chains.
We employ a fusion procedure similar to that in \cite{IMO12,IMO14} to construct Lax operators for skewsymmetric or symmetric representations, with which we build an integrable transfer matrix. 
The notation used in the paper closely follows that in \cite{GRW20}. 
We then proceed to diagonalise the transfer matrix using an approach based on a mixture of methods that appeared in \cite{DVK87,Res85,Res91} and more recently in \cite{GeRe20}.
We show that the problem of constructing transfer matrix eigenvectors may be reduced to an analogous problem for a $U_q(\mfgl_n)$-symmetric spin chain, which we solve inductively, following \cite{BeRa08}.
Our main results are Theorem \ref{T:sosp}, stating sufficient conditions for the Bethe vector to become an eigenvector of the transfer matrix as well as giving its corresponding eigenvalue, and Theorem \ref{T:tf}, stating a closed trace formula for the Bethe vectors, which generalises the formula given in \cite{TaVa13}.

The nesting procedure used in this paper relies on the chains of natural subalgebras of quantum loop algebras $U_q(\mfL\mfsp_{2n}) \supset U_q(\mfL\mfgl_{n}) \supset \cdots \supset U_q(\mfL\mfgl_{2})$ and $U_q(\mfL\mfso_{2n}) \supset U_q(\mfL\mfgl_{n}) \supset \cdots \supset U_q(\mfL\mfgl_{2})$. 
This gives an unambiguous solution to the spectral problem of a $U_q(\mfsp_{2n})$-symmetric spin chain. 
However, for a $U_q(\mfso_{2n})$-symmetric spin chain such a nesting breaks the underlying symmetry of the Dynkin diagram of type $D_n$ and additional steps must be taken to arrive at the correct Bethe equations, see Remark \ref{R:BEs} (ii) and (iii) for details.
This was not observed in \cite{DVK87} where a different choice of the nested vacuum vector was made and ultimately lead to a perturbative approach. This paper has resolved this issue. We believe that the results of this paper will be of interested to a wider community, in particular, in the study of the quantized Knizhnik-Zamolodchikov equations of type $C_n$ and $D_n$ in the spirit of \cite{Res92,TaVa94}.

The paper is organised as follows. In Section 2 we introduce the our notation and discuss relevant properties of $R$-matrices. In Section 3 we introduce the spin chain and discuss its structure, underlying symmetries and representations. Section 4 is the main section of this paper and contains our main results, Bethe vectors and Bethe equations. Appendix A gives a brief exposition of the quasi-classical limit.

{\it Acknowledgements.} The authors thank Niall MacKay for useful comments and Eric Ragoucy for suggesting to address this problem. 
A.\ G.\ thanks the Department of Mathematics, University of York for a Doctoral Prize Fellowship under its Mathematics Excellence Programme. 
V.\ R.\ thanks Department of Mathematics, University of York, for hospitality and support during the completion of this work.
This project has received funding from European Social Fund (project No 09.3.3-LMT-K-712-02-0017) under grant agreement with the Research Council of Lithuania (LMTLT). 
The authors gratefully acknowledge the financial support.

%%%%%%%%%%%%%%%%%%%%%%%%%%%%%%%%%%%%%%%%%%%%%%%%%%%%%%%%%%%%%%%%%%
% Section 2
%%%%%%%%%%%%%%%%%%%%%%%%%%%%%%%%%%%%%%%%%%%%%%%%%%%%%%%%%%%%%%%%%%

\section{Preliminaries and definitions}

%%%%%%%%%%%%%%%%%%%%%%%%%%%%%%%%%%%%%%%%%%%%%%%%%%%%%%%%%%%%%%%%%%

\subsection{Lie algebras and matrix operators}

Let $n\in\N$. We will denote by $\mfg_{2n}$ either the orthogonal Lie algebra $\mfso_{2n}$ or the symplectic Lie algebra $\mfsp_{2n}$. The Lie algebra $\mfg_{2n}$ can be realized as a Lie subalgebra of $\mfgl_{2n}$ as follows. For each $1\le i,j\le 2n$, let $E_{ij}$ denote the standard generators of $\mfgl_{2n}$ and put $\bi=2n-i+1$. Introduce elements $F_{ij} = E_{ij} - \theta_{ij} E_{\bj \bi}$ satisfying
\gat{ 
[F_{ij} , F_{kl} ] = \del_{jk} F_{il} - \del_{il} F_{kj} + \theta_{ij} ( \del_{j \bl} F_{k \bi} - \del_{i\bk} F_{\bj l} ) , \qq F_{ij} + \theta_{ij} F_{\bj\bi} = 0 , \label{[F,F]}
}
where $\theta_{ij}=\theta_i\theta_j$ with $\theta_i=-1$ if $1\le i\le n$ and $\theta_i=1$ if $n< i\le 2n$ in the symplectic case or $\theta_i=1$ in the orthogonal case. The algebra $\mfg_{2n}$ is isomorphic to $\spanc\{F_{ij} : 1\le i,j\le 2n \}$, and $\mathrm{span}_{\C} \{ F_{ii} \, : \, 1 \le i \le n \}$ forms a Cartan subalgebra which will be denoted by $\mfh_{2n}$. The elements $F_{ij}$ with $1\le i,j \le n$ form a subalgebra $\mfgl_n \subset \mfg_{2n}$; the same is true for the elements $F_{ij}$ with $n< i,j \le 2n$.

For convenience, we set $\theta=-1$ for $\mfsp_{2n}$ and $\theta=1$ for $\mfso_{2n}$, and $\theta' = \tfrac12(\theta+1)$. 
We will also need the following tuple of integers:
\equ{
(\nu_{1},\ldots,\nu_{2n}) = (-n + \theta',-n + 1 + \theta',\ldots,-1+\theta',1-\theta',\ldots,n-1-\theta',n-\theta'). \label{nu}
}

Let $M\in\N$. For any $1\le i,j\le M$ denote by $e_{ij}\in\End(\C^M)$ the usual matrix units over $\C$. For a matrix $X$ with entries $x_{ij}$ in an associative algebra $\mc{A}$ we write
\equ{
X_s = \sum_{1\le i,j\le M} I^{\ot s-1} \ot e_{ij} \ot I^{\ot k-s} \ot x_{ij} \in \End(\C^M)^{\ot k} \ot \mc{A} , \label{ten-op}
}
where $k\in \N_{\ge2}$ will always be clear from the context. 

For any matrix operator $A = \sum_{i,j=1}^n e_{ij} \ot a_{ij}$ with $e_{ij}\in \End(\C^n)$ and $a_{ij} \in \mc{A}$, any associative algebra, and any $1\le k\le n$, 
define a $k$-reduced operator $A^{(k)} := \sum_{i,j=k}^n e^{(k)}_{i-k+1,j-k+1} \ot a_{ij}$ with matrix units $e^{(k)}_{ij}\in\End(\C^{n-k+1})$. 
Given $l\ge1$ it will be convenient to say that $A^{(k+l)}$ is an $l$-reduced operator of $A^{(k)}$. 

%%%%%%%%%%%%%%%%%%%%%%%%%%%%%%%%%%%%%%%%%%%%%%%%%%%%%%%%%%%%%%%%%%

\subsection{Quantum $R$-matrices}

Choose $q\in\R^\times$, $q\ne \pm1$, and set $\ka=n-\theta$. Introduce a matrix-valued rational function $R$ by 
\equ{
R(u,v) := R_q + \frac{q-q^{-1}}{v/u-1}\,P - \frac{q-q^{-1}}{q^{2\ka}v/u-1}\,Q_q ,  \label{Ru}
}
where $R_q$, $P$ and $Q_q$ are matrix operators on $\C^{2n}\ot\C^{2n}$ defined by
\eqg{
R_q := \sum_{1\le i,j\le 2n}q^{\del_{ij}-\del_{i\bj}}e_{ii}\ot e_{jj} + (q-q^{-1}) \sum_{1\le i<j\le 2n} (e_{ij}\ot e_{ji} -  q^{\nu_i-\nu_j}\theta_{ij}e_{ij}\ot e_{\bi\bj}) ,
\\
P := \sum_{1\le i,j\le n} e_{ij} \ot e_{ji} , \qq Q_q := \sum_{1\le i,j\le 2n} q^{\nu_i-\nu_j} \theta_{ij} e_{ij} \ot e_{\bi\bj}.  \label{RPQ}
}
The matrix $R(u,v)$, obtained by Jimbo in \cite{Ji86}, is a solution of the quantum Yang-Baxter equation on $(\C^{2n})^{\ot 3}$ with spectral parameters,
\equ{ \label{YBE}
R_{12}(u,v)\,R_{13}(u,w)\,R_{23}(v,w) = R_{23}(v,w) \,R_{13}(u,w)\,R_{12}(u,v) ,
}
where we have employed the notation \eqref{ten-op}.

Let $k\ge l$. (We will not need the case when $k<l$.) With a slight abuse of notation we introduce ``reduced'' matrix operators $R^{(k,l)}_q$, $P^{(k,l)}$ %, $I^{(k,l)}_q$ 
and $Q^{(k,l)}_q$ on $\C^{n-k+1}\ot\C^{n-l+1}$ by
\eqg{ 
R^{(k,l)}_q := \sum_{i=1}^{n-k+1} \sum_{j=1}^{n-l+1} q^{\del_{ij}}e^{(k)}_{ii}\ot e^{(l)}_{jj} + (q-q^{-1}) \sum_{i,j=1}^{n-k+1} \del_{i<j}\,e^{(k)}_{ij}\ot e^{(l)}_{j'i'} , \\
P^{(k,l)} :=\sum_{i,j=1}^{n-k+1} e^{(k)}_{ij} \ot e^{(l)}_{j'i'}, \qq 
Q^{(k,l)}_q := \sum_{i,j=1}^{n-k+1} q^{i-j} e^{(k)}_{ij} \ot e^{(l)}_{\bi\bj} , \label{RPIQ}
}
where $i'=i+k-l$, $j'=j+k-l$ and $\bi = n-k-i+2$, $\bj = n-k-j+2$. The operators \eqref{RPQ} and \eqref{RPIQ} will never appear simultaneously, so there will be no ambiguity. We also note that
\equ{
(R^{(k,l)}_q)^{-1} = R^{(k,l)}_{q^{-1}} , \qq P^{(k,l)} Q^{(k,l)}_q P^{(k,l)} = Q^{(k,l)}_{q^{-1}} . \label{RQ}
}
Here the subscript $q^{-1}$ means that all instances of $q$ are replaced with $q^{-1}$ in the definition of the operator. 
Such notation will be used throughout this manuscript.

Recall that $\C^{2n} \cong \C^2 \ot \C^n$. Let $x_{ij}$ with $1\le i,j\le 2$ denote the matrix units of $\End(\C^2)$.
Then, for any $1\le i,j \le n$, we may write
\equ{
e_{ij} = x_{11} \ot e^{(1)}_{ij}, \qu e_{n+i,j} = x_{21} \ot e^{(1)}_{ij}, \qu e_{i,n+j} = x_{12} \ot e^{(1)}_{ij} ,\qu e_{n+i,n+j} = x_{22} \ot e^{(1)}_{ij} .
\label{e=x*e}
} 
Viewing the matrix $R(u,v)$ as an element in $\End(\C^2\ot \C^2)$ we recover the six-vertex block structure,
\eqa{
R(u,v) = \begin{pmatrix}  R^{(1,1)}(u,v)\! & \\ & \!K^{(1,1)}(u,q^{2\ka}v)\! & \!U^{(1,1)}(u,v)\! \\ & \!\bar{U}^{(1,1)}(v,u)\! & \!K^{(1,1)}(u,q^{2\ka}v)\! \\ &&& \!R^{(1,1)}(u,v) \end{pmatrix} , \label{R(u):new}
}
where the operators inside the matrix above are each acting on $\C^n \ot \C^n$ and are defined in terms of those in \eqref{RPIQ}, that is
\gat{
R^{(k,l)}(u,v) := R^{(k,l)}_q + \frac{q-q^{-1}}{v/u-1} \,P^{(k,l)}, \qq 
K^{(k,l)}(u,v) := (R^{(k,l)}_{q^{-1}}(u,v))^{\bar\om_2}, \label{RUK1} \\ 
U^{(k,k)}(u,v) := -\frac{q-q^{-1}}{u/v-1} \, P^{(k,k)} + \frac{\theta (q-q^{-1})}{q^{-\ka'}u/v-q^{\ka'}}\, Q^{(k,k)}_q , \qq 
\bar{U}^{(k,k)}(v,u) := P^{(k,k)} U_{q^{-1}}^{(k,k)}(v,u) P^{(k,k)} , \label{RUK2}
}
where $\ka'=\ka-k+1$ and $\bar\om$ is the transposition defined by $\bar{\om} : e^{(l)}_{ij} \mapsto q^{j-i} e^{(l)}_{\bj\,\bi}$, the inverse of which is $\om : e^{(l)}_{ij} \mapsto q^{i-j} e^{(l)}_{\bj\,\bi}$, and $\bi = n-i-l+1$, $\bj = n-j-l+1$. 
%
%
\if0
{ \red Matrix operators in \eqref{RUK} satisfy the following unitarity identities,
%
\equ{ 
R^{(k,l)}(u,v)\,R^{(k,l)}_{21}(v,u) = \frac{(q^2 u - v)(q^{-2} u - v)}{ (u-v)^2}\,I^{(k,l)} , \qq
K^{(k,l)}(u,v)\,K^{(k,l)}_{21}(v,q^{2(n-k+1)} u) = I^{(k,l)} , \label{unit}
}
where $R^{(k)}_{21}(v,u) = P^{(k)} R^{(k)}(v,u) P^{(k)}$ and $K^{(k)}_{21}(v,q^{2(n-k+1)}u)$ is defined similarly, and $I^{(k)}$ is the identity matrix on $\C^{n-k+1}\ot\C^{n-k+1}$.}
%
%
Matrices $R$ and $K$ in \eqref{RUK1} may be written as
%
\ali{
(R^{(k,l)}_{12}(u,v))^{t_2} &= C^{(k)}_1 K^{(k,l)}_{12}(u^{-1},v^{-1}) (C^{(k)}_1)^{-1} , \label{t-C-RK1} \\
(K^{(k,l)}_{12}(u,v))^{t_2} &= C^{(k)}_1 R^{(k,l)}_{12}(u^{-1},v^{-1}) (C^{(k)}_1)^{-1} , \label{t-C-RK2}
}
%
where $C^{(k)} := \sum_{i=1}^{n-k+1} q^{i} e^{(k)}_{i\bar\imath}$. 
\fi
We also note two more important identities,
\ali{
\bar{U}^{(k,k)}(v,u)\,(K^{(k,k)}(u,q^{2\ka}v))^{-1} &= \frac{q-q^{-1}}{v/u-1} P^{(k,k)}  (K^{(k,k)}(u,q^{2\ka}u))^{-1} \el &= \frac{1}{v-u}\Res{w\to u} R^{(k,k)}(u,w) (K^{(k,k)}(u,q^{2\ka}w))^{-1} , \label{UK=PK}
\\[.75em]
(K^{(k,k)}(v,q^{2\ka}u))^{-1} U^{(k,k)}(v,u) &= \frac{q^{-1}-q}{v/u-1} (K^{(k,k)}(u,q^{2\ka}u))^{-1} P^{(k,k)} \el &= \frac{1}{v-u}\Res{w\to u}(K^{(k,k)}(w,q^{2\ka}u))^{-1}  R^{(k,k)}(w,u) , 
\label{KU=KP}
}
that will play a key part in finding the so-called unwanted terms of the algebraic Bethe ansatz. 
Lastly, introduce elements
\gat{
\mc{E}^{+(l)}_{ij} := \del_{ij} q^{-e^{(l)}_{ii}} - (q-q^{-1}) \del_{i<j} e^{(l)}_{ij} , \qq \mc{E}^{-(l)}_{ji} := \del_{ij} q^{e^{(l)}_{ii}} + (q-q^{-1}) \del_{i<j} e^{(l)}_{ji} , \label{E1} \\
\mc{E}^{(l)}_{ij}(u) := \frac{1}{1-u}\,\mc{E}^{+(l)}_{ij} + \frac{1}{1-u^{-1}} \,\mc{E}^{-(l)}_{ij} , \label{E2}
}
where we have used the notation $q^{e^{(l)}_{ii}} = \sum_{j=1}^{n-l+1} q^{\del_{ij}} e^{(l)}_{jj}$. Then we may write
\ali{
R^{(k,l)}(u,v) = \sum_{i,j=1}^{n-k+1} e^{(k)}_{ij} \ot \mc{E}^{(l)}_{j' i'}(v/u) , 
\qq
K^{(k,l)}(u,v) = \sum_{i,j=1}^{n-k+1} e^{(k)}_{ij} \ot (\mc{E}^{(l)}_{q^{-1},j'i'}(v/u))^{\bar\om} , \label{RK-factor}
}
where $i'=i+k-l$ and $j'=j+k-l$.

%%%%%%%%%%%%%%%%%%%%%%%%%%%%%%%%%%%%%%%%%%%%%%%%%%%%%%%%%%%%%%%%%%
% Section 3
%%%%%%%%%%%%%%%%%%%%%%%%%%%%%%%%%%%%%%%%%%%%%%%%%%%%%%%%%%%%%%%%%%

\section{Setting up symmetries and representations of the spin chain}

%%%%%%%%%%%%%%%%%%%%%%%%%%%%%%%%%%%%%%%%%%%%%%%%%%%%%%%%%%%%%%%%%%
%%%%%%%%%%%%%%%%%%%%%%%%%%%%%%%%%%%%%%%%%%%%%%%%%%%%%%%%%%%%%%%%%%

\subsection{Quantum loop algebras $U_q(\mf{L}\mfg_{2n})$ and $U_q(\mf{L}\mfgl_{n})$} \label{sec:Uq}

Let $N=2n$ or $N=n$. Introduce elements $\ell^\pm_{ij}[r]$ with $1\leq i,j\leq N$ and $r\geq0$, combine them into formal series $\ell^\pm_{ij}(u)= \sum_{r\ge0} \ell^\pm_{ij}[r]\, u^{\pm r}$, and collect into generating matrices
\equ{
L^\pm(u) := \sum_{1\le i,j\le N} e_{ij} \ot \ell^\pm_{ij}(u). \label{L(u)}
}
We will say that that elements $\ell^\pm_{ij}[r]$ have degree $r$.

\begin{defn} \label{D:Ugn:aff}
The quantum loop algebra $U_q(\mf{L}\mfg_{2n})$ (resp.\ $U_q(\mf{L}\mfgl_{n})$) is the unital associative algebra with generators $\ell^\pm_{ij}[r]$ with $1\leq i,j\leq N$ and $r\geq0$, subject to the following relations:
\gat{
\ell^-_{ii}[0]\,\ell^+_{ii}[0] = 1 \;\text{ for all $i$ and }\; 
\ell^-_{ij}[0] = \ell^+_{ji}[0] = 0 \;\text{ for $i<j$ and} \\
R_{12}(u,v)\,L^\pm_1(u)\,L^\pm_2(v) = L^\pm_2(v)\,L^\pm_1(u)\,R_{12}(u,v), \label{RLL1}\\
R_{12}(u,v)\,L^+_1(u)\,L^-_2(v) = L^-_2(v)\,L^+_1(u)\,R_{12}(u,v) . \label{RLL2}
}
where $N=2n$ and $R_{12}(u,v)$ is given by \eqref{Ru} (resp.\ $N=n$ and $R_{12}(u,v)$ is given by \eqref{RUK1}).
\end{defn}

The following subalgebras of $U_q(\mfL\mfg_{2n})$ will be relevent to the present work:
\begin{itemize}

\item The subalgebra generated by $\lpm{ij}[0]$ with $1\le i,j\le 2n$, isomorphic to the quantum enveloping algebra $U_q(\mfg_{2n})$ of the direct sum of the Lie algebra $\mfg_{2n}$ and a one dimensional Lie algebra. \smallskip

\item The subalgebra generated by $\lpm{ij}[r]$ with $1\le i,j\le n$ and $r\ge 0$, isomorphic to $U_q(\mfL\mfgl_{n})$. \smallskip

\item The subalgebra generated by $\lpm{ij}[0]$ with $1\le i,j\le n$, isomorphic to the quantum enveloping algebra $U_q(\mfgl_{n})$ of the Lie algebra $\mfgl_{n}$.

\end{itemize}

We write the generating matrices $L^\pm(u)$ of $U_q(\mfL\mfg_{2n})$ as
\equ{
L^\pm(u) = \begin{pmatrix} A^{\pm}(u)\! & \!B^{\pm}(u) \\ C^{\pm}(u)\! & \!D^{\pm}(u) \end{pmatrix} .  \label{Lu:new} 
}
Then viewing $L^\pm_1(u)$ and $L^\pm_2(u)$ as elements in $\End(\C^2\ot \C^2)$ with entries in $\End(\C^n \ot \C^n) \ot U^{ex}_q(\mfL\mfg_{2n})[[u^{\pm1}]]$, we have
\[
L^\pm_1(u) = \begin{pmatrix} A^{\pm}_1(u)\! & & \!B^{\pm}_1(u)\! \\ & \!A^{\pm}_1(u)\! && \!B^{\pm}_1(u) \\ C^{\pm}_1(u)\! & & \!D^{\pm}_1(u) \\ & \!C^{\pm}_1(u)\! && \!D^{\pm}_1(u) \end{pmatrix} \!, \qq
L^\pm_2(u) = \begin{pmatrix} A^{\pm}_2(u)\! & \!B^{\pm}_2(u)\! \\ C^{\pm}_2(u)\! & \!D^{\pm}_2(u)\! \\ && \!A^{\pm}_2(u)\! & \!B^{\pm}_2(u) \\ && \!C^{\pm}_2(u)\! & \!D^{\pm}_2(u) \end{pmatrix} \!.  
\]
This allows us to write the defining relations of $U_q(\mfL\mfg_{2n})$ in terms of the matrix operators $A^{\pm}(u)$, $B^{\pm}(u)$, $C^{\pm}(u)$ and $D^{\pm}(u)$. The relations that we will need are: 
\ali{
A^{\pm}_2(v)\, B^{\pm}_1(u)\, K^{(1,1)}_{12}(u,q^{2\ka}v) &= R^{(1,1)}_{12}(u,v)\, B^{\pm}_1(u) A^{\pm}_2(v) - B^{\pm}_2(v)\, A^{\pm}_1(u) \,\bar{U}^{(1,1)}_{12}(v,u) ,  \label{AB}
\\
K^{(1,1)}_{12}(v,q^{2\ka}u)\,D^{\pm}_1(v)\,B^{\pm}_2(u) &= B^{\pm}_2(u)\, D^{\pm}_1(v)\, R^{(1,1)}_{12}(v,u) - U^{(1,1)}_{12}(v,u)\, B^{\pm}_1(v)\, D^{\pm}_2(u) , \label{DB} \\
K^{(1,1)}_{12}(u,q^{2\ka}v)\, C^\pm_1(u)\, A^\pm_2(v) &= A^\pm_2(v)\, C^\pm_1(u)\, R^{(1,1)}_{12}(u,v) - U^{(1,1)}_{12}(u,v)\, A^\pm_1(u)\, C^\pm_2(v) \label{CA} , \\
C^\pm_2(v)\, D^\pm_2(u)\,K^{(1,1)}_{12}(u,q^{2\ka}v) &= R^{(1,1)}_{12}(u,v)\,D^\pm_1(u)\, C^\pm_2(v) - D^\pm_2(v)\, C^\pm_1(u)\,\bar{U}^{(1,1)}_{12}(v,u) , \label{CD} \\
K^{(1,1)}_{12}(u,q^{2\ka}v)\, D^\pm_1(u)\, A^\pm_2(v) & - A^\pm_2(v)\, D^\pm_1(u)\, K^{(1,1)}_{12}(u,q^{2\ka}v) \el &= B^\pm_2(v)\,C^\pm_1(u)\, \bar{U}^{(1,1)}_{12}(v,u) - U^{(1,1)}_{12}(u,v)\, B^\pm_1(u)\, C^\pm_2(v) , 
\label{AD}
}
and their mixed counterparts obtained in an obvious way, cf.~\eqref{RLL1} vs.~\eqref{RLL2}. Operators $A^{\pm}(u)$, $B^{\pm}(u)$ and $D^{\pm}(u)$ satisfy relations analogous to \eqref{RLL1} and \eqref{RLL2} only with $R^{(1,1)}(u,v)$, e.g.,
\equ{
R^{(1,1)}_{12}(u,v)\,B^{\pm}_1(u)\,B^{\pm}_2(v) = B^{\pm}_2(v)\,B^{\pm}_1(u)\,R^{(1,1)}_{12}(u,v) . \label{BB}
}

We now focus on the subalgebra $U_q(\mfgl_n)\subset U_q(\mfg_{2n})$ generated by coefficients of the matrix entries of $A^{\pm}(u)$. Define a $k$-reduced matrix $A^{\pm(k)}(u) := \sum_{i,j=k}^n e^{(k)}_{i-k+1,j-k+1} \ot [A^{\pm}(u)]_{ij}$ and set
\equ{
\mr{a}^{\pm(k)}(v) := [A^{\pm(k)}(v)]_{11} , \qq 
B^{\pm(k+1)}(u) := \sum_{j=1}^{n-k} (e^{(k+1)}_j)^* \ot [A^{\pm(k)}(u)]_{1,1+j} . \label{def-a-B}
}
We also define a suitably normalised check $\check{R}$-matrix
\[
\check{R}^{(k,l)}_{12}(u,v) := \frac{v-u}{q v-q^{-1} u}\,P^{(k,l)}_{12} R^{(k,l)}_{12}(u,v) .
\]
The defining relations of $U_q(\mfL\mfgl_n)$ then yield
\gat{
\mr{a}^{\pm(k)}(v)\, B^{\pm(k+1)}_1(u) = \frac{q\,v-q^{-1}u}{v-u}\,B^{\pm(k+1)}_1(u) \,\mr{a}^{\pm(k)}(v) + \frac{q-q^{-1}}{u/v-1} \,B^{\pm(k+1)}_1(v) \,\mr{a}^{\pm(k)}(u) , 
 \label{gln:ab} \\
A^{\pm(k)}_1(v)\,B^{\pm(k)}_2(u) = B^{\pm(k)}_2(u)\,A^{\pm(k)}_1(v)\,R^{(k,k)}_{12}(v,u) - \frac{q-q^{-1}}{u/v-1}\,B^{\pm(k)}_2(v)\, A^{\pm(k)}_1(u)\,P^{(k,k)}_{12}, \label{gln:db} \\
B^{\pm(k)}_1(u)\,B^{\pm(k)}_2(v) = B^{\pm(k)}_1(v)\,B^{\pm(k)}_2(u)\,\check{R}^{(k,k)}_{12}(u,v) , \label{gln:bb} \\[.5em]
R^{(k,k)}_{12}(u,v)\,A^{\pm(k)}_1(u)\,A^{\pm(k)}_2(v) = A^{\pm(k)}_2(v)\,A^{\pm(k)}_1(u)\,R^{(k,k)}_{12}(u,v) , \label{gln:dd} 
}
plus the mixed relations.
Note that $R^{(n,n)}_{12}(u,v)$ acts as a constant on $\C\ot \C \cong \C$ multiplying by $\frac{qv-q^{-1}u}{v-u}$, while $\check{R}^{(n,n)}_{12}(u,v)$ multiplies by $1$.

%%%%%%%%%%%%%%%%%%%%%%%%%%%%%%%%%%%%%%%%%%%%%%%%%%%%%%%%%%%%%%%%%%
%%%%%%%%%%%%%%%%%%%%%%%%%%%%%%%%%%%%%%%%%%%%%%%%%%%%%%%%%%%%%%%%%% 

\subsection{Representations} \label{sec:reps}  

Fix $\ell \in \N$, the length of a spin chain, and choose $c_1,\ldots,c_\ell \in \C^\times$, called inhomogeneities or marked points. Then choose $s_1,\ldots,s_\ell \in \N$ such that $1\le s_i\le n$ for all $i$ in the symplectic case, and $s_i\ge 1$ for all $i$ in the orthogonal case. Our goal is to study a spectral problem in the space
\equ{
L := L(\la^{(1)})_{c_1} \ot \cdots \ot L(\la^{(\ell)})_{c_\ell}  \label{L}
}
where $L(\la^{(i)})_{c_i}$ with $i=1,\ldots,\ell$ denote the (skew)symmetric $U_q(\mfL\mfg_{2n})$-modules of lowest weight $\la^{(i)}(u)$ with components given by
\ali{
\la^{(i)}_j(u,c_i) &:= \begin{cases} 
\dfrac{q^{s_i}\,c_i - q^{-s_i} u}{c_i - u}  &\text{if}\qu j=1, \\[.75em]
1 &\text{if}\qu 1< j < 2n, \\[.25em] 
\dfrac{q^{-1}c_i - q^{-2\ka+1}\,u}{q^{s_i-1}c_i - q^{-2\ka-s_i+1}u} & \text{if}\qu j=2n  \end{cases} \label{la-symm}
\intertext{in the symmetric case, i.e.\ when $\mfg_{2n} = \mfso_{2n}$, and by}
\la^{(i)}_j(u,c_i) &:= \begin{cases} 
\dfrac{q\,c_i - q^{-1} u}{c_i - u}  &\text{if}\qu 1\le j \le s_i, \\[.75em]
1 &\text{if}\qu s_i< j < 2n-s_i+1, \\[.25em] 
\dfrac{q^{-s_i} c_i - q^{-2\ka+s_i} u}{q^{-s_i+1}\,c_i - q^{-2\ka+s_i-1}u} & \text{if}\qu 2n-s_i+1\le j\le 2n  \end{cases} \label{la-skew}
}
in the skewsymmetric case, i.e.\ when $\mfg_{2n} = \mfsp_{2n}$. In particular, given a lowest weight vector $\eta \in L$ we have
\equ{
\lpm{ij}(u)\,\eta = 0 \;\text{ for }\;i>j\;\text{ and }\;  
\lpm{jj}(u)\,\eta = \prod_{i=1}^\ell \la^{(i)}(u,c_i) \, \eta \;\text{ for all }j. \label{hw}
}
Here weights $\la^{(i)}(u)$ should be expanded in positive (resp.\ negative) series in $u$ for $\lp{jj}(u)$ (resp.\ for $\lm{jj}(u)$).

As vector spaces, modules $L(\la^{(i)})_{c_i}$ are isomorphic to the subspaces $\Pi^\theta_{s_i} (\C^{2n})^{\ot s_i} \subset (\C^{2n})^{\ot s_i}$ where $\Pi^\pm_{s_i}$ are idempotent operators defined by $\Pi^\pm_1 := 1$ and, when $s_i\ge 2$,
\[
\Pi^{\theta}_{s_i} := \frac{\theta}{[s_i]_q!}\prod_{j=s_i}^2 \Big( R_{12}(q^{-2\theta},1 )P_{12} \cdots R_{j-1,j}(q^{-(j-1)\theta},1)P_{j-1,j} \Big) .
\]
The lowest weight vector of $L(\la^{(i)})_{c_i}$ is $\eta_{s_i} = e_1 \ot \cdots \ot e_1$ in the symmetric case, and it is $\eta_{s_i} = \sum_{\si\in \mf{S}_{s_i}} \sign(\si) \,q^{l(\si)} \cdot e_{\si(1)} \ot e_{\si(2)} \ot \cdots \ot e_{\si(s_i)}$ in the skewsymmetric case;
here $l(\si)$ denotes the length of a reduced expression of $\si\in\mf{S}_{s_i}$, an element in the symmetric group on $s_i$ letters (full details will be given in \cite{GRW20}).
The generating matrix $L^\pm_a(u)$ acts on the space $L$ in terms of a product of $R$-matrices \eqref{Ru},
\equ{
T_a(u; \bm c) := \prod_{i=1}^\ell  \prod_{j=1}^{s_i} R_{ai_j}(u,q^{2\theta(j-1)}c_i)  \label{mono}
}
where $i_j$ enumerate individual tensorands $\C^{2n}$ of $L(\la^{(i)})_{c_i}$. We will often omit the depence on $\bm c$ to ease the notation and write $T_a(u)$, its matrix elements will be denoted as $t_{ij}(u)$.

Crucially, modules $L(\la^{(i)})_{c_i}$ are irreducible $U_q(\mfg_{2n})$-modules of lowest weight $\la^{(i)} = \left( q^{s_i}, 1, \ldots, 1 , q^{-s_i} \right)$ in the symmetric case, and of weight $\la^{(i)} = \left( q, \ldots , q, 1, \ldots, 1 , q^{-1}, \ldots, q^{-1} \right)$ in the skewsymmetric case; here the number of $q$'s and $q^{-1}$'s is $s_i$. The subspace
\[
(L(\la^{(i)})_{c_i})^0 :=  \{ \xi \in L(\la^{(i)})_{c_i} \,:\, t_{n+i,j}[0]\,\xi = 0 \;\text{for}\; 1\le i,j\le n \}
\]
is an irreducible $U_q(\mfgl_{n})$-module of lowest weight $(\la^{(i)}_1, \ldots, \la^{(i)}_n)$. In particular, we have that
\[
L^0 :=  \{ \xi \in L \,:\, t_{n+i,j}(u)\,\xi = 0 \;\text{for}\; 1\le i,j\le n \} = L(\la^{(1)})_{c_1}^0 \ot \cdots \ot L(\la^{(\ell)})_{c_\ell}^0 .
\]
The $A^\pm_a(u)$ and $D^\pm_a(u)$ operators act on the subspace $L^0$ in terms of a product of the ``reduced'' $R$- and $K$-matrices defined in \eqref{RUK1},
\ali{
A^{(1)}_a(u) & := \prod_{i=1}^\ell  \prod_{j=1}^{s_i} R^{(1,1)}_{ai_j}(u,q^{2\theta(j-1)}c_i) , \label{amono} \\
D^{(1)}_a(u) & := \prod_{i=1}^\ell  \prod_{j=1}^{s_i} K^{(1,1)}_{ai_j}(u,q^{2\ka+2\theta(j-1)}c_i) . \label{dmono}
}

%%%%%%%%%%%%%%%%%%%%%%%%%%%%%%%%%%%%%%%%%%%%%%%%%%%%%%%%%%%%%%%%%%
% Section 4
%%%%%%%%%%%%%%%%%%%%%%%%%%%%%%%%%%%%%%%%%%%%%%%%%%%%%%%%%%%%%%%%%%

\section{Algebraic Bethe ansatz}

%%%%%%%%%%%%%%%%%%%%%%%%%%%%%%%%%%%%%%%%%%%%%%%%%%%%%%%%%%%%%%%%%%
%%%%%%%%%%%%%%%%%%%%%%%%%%%%%%%%%%%%%%%%%%%%%%%%%%%%%%%%%%%%%%%%%%

\subsection{Quantum spaces and monodromy matrices} \label{sec:mono}

Choose $m_0,m_1,\dots,m_{n-1}\in\Z_{\ge0}$, which we call excitation or magnon numbers. For each $m_k$ assign an $m_k$-tuple $\bm u^{(k)} := (u^{(k)}_1,\dots,u^{(k)}_{m_k})$ of complex parameters and an $m_k$-tuple of labels $\bm a^{k} := ( a^k_1, \dots, a^k_{m_k} )$. For $m_0$ we additionally assign a tuple $\tl{\bm a}^{0} := ( \ta^0_1, \dots, \ta^0_{m_0} )$. We will often use the following shorthand notation:
\equ{
\bm u^{(0\dots k)} := (\bm u^{(0)},\dots,\bm u^{(k)}), \qq 
\bm a^{\tl 0,0\dots k} := (\tl{\bm a}^{0},\bm a^{0},\dots,\bm a^{k}) . \label{multi}
}
Let $V^{(k)}_{a^{k-1}_i}$ denote a copy of $\C^{n-k+1}$ labelled by ``$a^{k-1}_i$'' and let $W^{(k)}_{\bm a^{k-1}}$ be given by
\[
W^{(k)}_{\bm a^{k-1}} := V^{(k)}_{a^{k-1}_1} \ot \cdots \ot V^{(k)}_{a^{k-1}_{m_{k-1}}} .
\]

Let $L$ be a lowest weight $U_q(\mfL\mfg_{2n})$-module defined in \eqref{L}. We will say that $L^{(0)} := L$ is a level-$0$ quantum space. We define a level-$1$ quantum space by
\equ{
L^{(1)} := (L^{(0)})^0 \ot W^{(1)}_{\tl{\bm a}^{0}} \ot W^{(1)}_{\bm a^{0}} . \label{L1}
}
Then for each $2\le k \le n$ we recursively define a level-$k$ quantum space by
\equ{
L^{(k)} := (L^{(k-1)})^0 \ot W^{(k)}_{\bm a^{k-1}}  \label{Lk}
}
where
\[
(L^{(k-1)})^0 := \{\xi \in L^{(k-1)} \,:\, t_{ij}(u)\,\xi = 0 \text{ for } i>j \text{ and } j < k-1 \}.
\]

\begin{defn} \label{D:mono}
We will say that $T_a(v)$ is a level-0 monodromy matrix. We define level-$1$ monodromy matrices, acting on the space $L^{(1)}$, by 
\ali{
A^{(1)}_{a\bm a^{\tl 0,0}}(v;\bm u^{(0)}) & := \left(\prod_{i=m_0}^{1} K^{(1,1)}_{a \tl a^{0}_i}(v,q^{2\theta} u^{(0)}_i) \right) \left(\prod_{i=1}^{m_0} K^{(1,1)}_{aa^{0}_i}(v,u^{(0)}_i) \right) A^{(1)}_a(v)  , \label{A1} \\
D^{(1)}_{a\bm a^{\tl 0,0}}(v;\bm u^{(0)}) & := D^{(1)}_a(v) \left(\prod_{i=m_0}^{1} R^{(1,1)}_{a\tl a^{0}_i}(v, u^{(0)}_i) \right) \left(\prod_{i=1}^{m_0} R^{(1,1)}_{aa^{0}_i}(v,q^{-2\theta} u^{(0)}_i) \right) . \label{D1}
}
For each $2\le k \le n$ we recursively define level-$k$ monodromy matrices, acting on the spaces $L^{(k)}$, by
\ali{ 
A^{(k)}_{a\bm a^{\tl 0,0\dots k-1}}(v;\bm u^{(0\dots k-1)}) & := A^{(k)}_{a\bm a^{\tl 0,0\dots k-2}}(v;\bm u^{(0\dots k-2)}) \left(\prod_{i=1}^{m_{k-1}} R^{(k,k)}_{aa^{k-1}_i}(v,u^{(k-1)}_i) \right) , \label{Ak} 
\\
D^{(k)}_{a\bm a^{\tl 0,0\dots k-1}}(v;\bm u^{(0\dots k-1)}) & := \left(\prod_{i=1}^{m_{k-1}} R^{(k,k)}_{aa^{k-1}_i}(v,q^{2\ka'}u^{(k-1)}_i) \right)  D^{(k)}_{a\bm a^{\tl 0,0\dots k-2}}(v;\bm u^{(0\dots k-2)}) , \label{Dk}
}
where $A^{(k)}_{a {\bm a}^{\tl 0,0\dots k-2}}$ and $D^{(k)}_{a{\bm a}^{\tl 0,0\dots k-2}}$ 
denote the 1-reduced operators of $A^{(k-1)}_{a{\bm a}^{\tl0,0\dots k-2}}$ and $D^{(k-1)}_{a{\bm a}^{\tl 0,0\dots k-2}}$, respectively.
\end{defn}

Operators $A^{(k)}_{a\bm a^{\tl 0,0\dots k-1}}$ and $D^{(k)}_{a\bm a^{\tl 0,0\dots k-1}}$ 
are matrices with entries in $\End(L^{(k)})$. Thus, to ease the notation, we will write them as $A^{(k)}_{a}$ and $D^{(k)}_{a}$. We will use a similar notation throughout the manuscript.

\begin{lemma} \label{L:a,d-rll}
For $1\le k\le n-1$ let $\equiv$ denote equality of operators in the space $L^{(k)}$. 
Then
\ali{
& R^{(k,k)}_{ab}(v,w)\,A^{(k)}_{a}(v;\bm u^{(0\dots k-1)})\,A^{(k)}_{b}(w;\bm u^{(0\dots k-1)}) \nn\\
& \qq \equiv A^{(k)}_{b}(w;\bm u^{(0\dots k-1)})\,A^{(k)}_{a}(v;\bm u^{(0\dots k-1)})\,R^{(k,k)}_{ab}(v,w) , \nn\\[.5em]
& R^{(k,k)}_{ab}(v,w)\,D^{(k)}_{a}(v;\bm u^{(0\dots k-1)})\,D^{(k)}_{b}(w;\bm u^{(0\dots k-1)}) \nn\\
& \qq \equiv D^{(k)}_{b}(w;\bm u^{(0\dots k-1)})\,D^{(k)}_{a}(v;\bm u^{(0\dots k-1)})\,R^{(k,k)}_{ab}(v,w) , \nn\\[.5em]
& D^{(k)}_{a}(v;\bm u^{(0\dots k-1)})\,R^{(k,k)}_{q^{-1},ab}(v,q^{2\ka'}w)\,A^{(k)}_{b}(w;\bm u^{(0\dots k-1)}) \nn\\
& \qq \equiv A^{(k)}_{b}(w;\bm u^{(0\dots k-1)})\,R^{(k,k)}_{q^{-1},ab}(v,q^{2\ka'}w)\,D^{(k)}_{a}(v;\bm u^{(0\dots k-1)})   \label{DRA=ARD}
}
where $\ka'=\ka-k+1$.
\end{lemma}

\begin{proof}
The first two identities follow from the Yang-Baxter equation and the defining relations of $A$ and $D$ operators. For the third identity we additionally need to use the property $C^\pm_a(v) \equiv 0$.
\end{proof}

%%% level-k lowest weight vector

Observe that $K$-matrices in \eqref{A1} and $R$-matrices in \eqref{D1} are ``at the fusion point''. More precisely, introduce (skew)symmetric projector $\Pi^\pm := \pm \frac{1}{q+q^{-1}}R^{(1,1)}(q^{\mp2},1) P^{(1,1)}$ and set $V^\pm := \Pi^\pm\, \C^n \ot \C^n$. The subspace $V^\pm$ is an irreducible $U_q(\mfL\mfgl_n)$-module with the lowest weight vector $\xi^\pm$ given by 
\equ{
\xi^+ = e^{(1)}_1 \ot e^{(1)}_1 , \qq  \xi^- = e^{(1)}_1 \ot e^{(1)}_2 - q\, e^{(1)}_2 \ot e^{(1)}_1 . \label{xi}
}
Denote
\[
\mr{K}^\theta_{kj} := \big[K^{(1,1)}_{a\tl a_i^0}(v, q^{2\theta} u^{(0)}_i)\,K^{(1,1)}_{aa_i^0}(v, u^{(0)}_i)\big]_{jk} , \qq
\mr{R}^\theta_{kj} := \big[R^{(1,1)}_{a\tl a_i^0}(v, u^{(0)}_i)\,R^{(1,1)}_{aa_i^0}(v, q^{-2\theta} u^{(0)}_i)\big]_{jk} ,
\]
where the matrix elements are taken with respect to the ``$a$'' space. 
Then $\mr{K}^\theta_{jk} \,\xi^{-\theta} = \mr{R}^\theta_{jk} \,\xi^{-\theta} = 0$ if $j>k$ and
\ali{
\mr{K}^-_{jj} \,\xi^{+} &= \bigg( \del_{j<n} + \del_{jn} \, \frac{q^2 v - q^{-2} u^{(0)}_i}{v - u^{(0)}_i} \bigg) \, \xi^{+} , \qq &
\mr{R}^-_{jj} \,\xi^{+} &= \bigg( \del_{j1}\, \frac{q^{-2} v - q^2 u^{(0)}_i}{v - u^{(0)}_i} + \del_{j>1} \bigg) \, \xi^{+} ,
\\
\mr{K}^+_{jj} \,\xi^{-} &= \bigg( \del_{j<n-1} + \del_{j\ge n-1}\,\frac{q\,v - q^{-1} u^{(0)}_i}{v - u^{(0)}_i} \bigg)\, \xi^{-} , \qq &
\mr{R}^+_{jj} \,\xi^{-} &= \bigg( \del_{j\le 2}\, \frac{q^{-1} v - q\, u^{(0)}_i}{v - u^{(0)}_i}  + \del_{j>2} \bigg)\, \xi^{-} .
}

For each $1\le j \le m_0$ define vector $\xi^{(j)}_-$ recursively by $\xi^{(1)}_- := \xi^-$ and
\[
\xi^{(j)}_- := e^{(1)}_1 \ot \xi^{(j-1)}_- \ot e^{(1)}_2 - q\, e^{(1)}_2 \ot \xi^{(j-1)}_- \ot e^{(1)}_1 .
\]
We also set $\xi^{(j)}_+ := \big(e^{(1)}_1\big)^{\ot 2j}$.
Then for each $1\le k \le n-1$ we define a \emph{level-$k$ nested vacuum vector} by
\equ{
\eta^{(k)}_\pm := \eta \ot \xi^{(m_0)}_{\pm} \ot \big(e^{(2)}_1\big)^{\ot m_{1}} \ot \cdots \ot \big(e^{(k+1)}_1\big)^{\ot m_{k}} \in (L^{(k)})^{0} , \label{eta-k}
}
where $\eta = \eta_{s_1} \ot \cdots \ot \eta_{s_\ell}$ is the lowest weight vector of $L^{(0)}= L$.
%
%%% action of a-operators
%
We then denote the $(1,1)$-th matrix element of monodromy matrices \eqref{A1}, \eqref{Ak} by 
\[
\mr{a}^{(k)}(v;\bm u^{(0\dots k-1)}) := \big[ A^{(k)}_a(v;\bm u^{(0\dots k-1)}) \big]_{11} , \qq
\mr{d}^{(k)}(v;\bm u^{(0\dots k-1)}) := \big[ D^{(k)}_a(v;\bm u^{(0\dots k-1)}) \big]_{11} .
\]
We will be interested in the action of these operators on $\eta^{(k)}_\pm$.

\begin{lemma} \label{L:a-action}
Vector $\eta^{(k)}_{-\theta}$ is lowest weight vector with respect to the action of the level-$k$ monodromy matrix. 
The operators $\mr{a}^{(k)}(v;\bm u^{(0\dots k-1)})$ and $\mr{d}^{(k)}(v;\bm u^{(0\dots k-1)})$ act on $\eta^{(k)}_{-\theta}$ by multiplication with 
\ali{
\text{for } k=1: \qu & \prod_{i=1}^\ell\la^{(i)}_{1}(v) \,, \label{a1}
\\
\text{for } 2\le k\le n-2: \qu & \prod_{i=1}^\ell\la^{(i)}_{k}(v) \prod_{i=1}^{m_{k-1}} \frac{q^{-1} v-q\,u^{(k-1)}_i}{v-u^{(k-1)}_i}   \,, \label{ak}
\\
\text{for } k= n-1: \qu & \prod_{i=1}^\ell\la^{(i)}_{n-1}(v) \prod_{i=1}^{m_{n-2}} \frac{q^{-1} v-q\,u^{(n-2)}_i}{v-u^{(n-2)}_i} \,  \prod_{i=1}^{m_0}\frac{q^{\theta'}\,v-q^{-\theta'}u^{(0)}_i}{v-u^{(0)}_i}  \,, \label{an-1}
\\
\text{for } k=n: \qu & \prod_{i=1}^\ell\la^{(i)}_{n}(v) \prod_{i=1}^{m_{n-1}} \frac{q^{-1} v-q\,u^{(n-1)}_i}{v-u^{(n-1)}_i}  \prod_{i=1}^{m_0} \frac{q^{2-\theta'}\,v-q^{\theta'-2}u^{(0)}_i}{v-u^{(0)}_i}  \,, \label{an}
\intertext{and} 
\text{for } k=1: \qu & \prod_{i=1}^\ell \la^{(i)}_{2n}(v)  \,, \label{d1}
\\
\text{for } 2\le k\le n-2: \qu & \prod_{i=1}^\ell \la^{(i)}_{2n-k+1}(v) \prod_{i=1}^{m_{k-1}} \frac{q\,v-q^{2\ka'-1} u^{(k-1)}_i}{v-q^{2\ka'}u^{(k-1)}_i}  \,, \label{dk}
\\
\text{for } k= n-1: \qu & \prod_{i=1}^\ell \la^{(i)}_{n+2}(v) \prod_{i=1}^{m_{n-2}} \frac{q\,v-q^{2\ka'-1} u^{(n-2)}_i}{v-q^{2\ka'}u^{(n-2)}_i} \,  \prod_{i=1}^{m_0}\frac{q^{-\theta'} v-q^{\theta'}\,u^{(0)}_i}{v-u^{(0)}_i}  \,, \label{dn-1}
\\
\text{for } k=n: \qu & \prod_{i=1}^\ell \la^{(i)}_{n+1}(v) \prod_{i=1}^{m_{n-1}} \frac{q\,v-q^{2\ka'-1} u^{(n-1)}_i}{v-q^{2\ka'}u^{(n-1)}_i} \prod_{i=1}^{m_0} \frac{q^{\theta'-2}\,v-q^{2-\theta'}u^{(0)}_i}{v-u^{(0)}_i}  \,, \label{dn}
}
respectively. Here $\ka'=\ka-k+1$ and $\theta' = \tfrac12(1+\theta)$.
\end{lemma}

\begin{proof} 
First, note that the level-$k$ nested monodromy matrix may be written as
\aln{
A^{(k)}_{a}(v;\bm u^{(0\dots k-1)}) & = \Bigg(\prod_{i=m_0}^{1} K^{(k,1)}_{a \tl a^{0}_i}(v,q^{2\theta} u^{(0)}_i) \Bigg) \Bigg(\prod_{i=1}^{m_0} K^{(k,1)}_{a a^{0}_i}(v, u^{(0)}_i) \Bigg) A^{(k)}_a(v)  \Bigg(\prod_{l=2}^{k} \prod_{i=1}^{m_{l-1}} R^{(k,l)}_{aa^{l-1}_i}(v,u^{(l-1)}_i) \Bigg) .
}
In order to prove that $\eta_{-\theta}^{(k)}$ is a lowest weight vector, consider the action of the elements $\big[A^{(k)}_{a}(v;\bm u^{(0\dots k-1)}\big]_{ij}$ with $i\geq j$.
It follows from (\ref{E1}--\ref{RK-factor}) that, when acting on $\eta^{(k)}_{+}$, the $R$- and $K$-matrices become upper triangular matrices in the ``$a$'' auxiliary space.
That is, for $i\geq j$,
\aln{
\mc{E}^{(l)}_{j'i'}(u/v)\, e^{(l)}_1 &= \del_{ij} \bigg(\del_{i'1}\, \frac{q^{-1} v-q\,u}{v-u} + \del_{i'>1} \bigg) e^{(l)}_1, \el
(\mc{E}^{(l)}_{q^{-1},j'i'}(u/v))^{\bar\om}\, e^{(l)}_1 &= \del_{ij} \bigg( \del_{i' \bar l}\,\, \frac{q\,v-q^{-1}u}{v-u} + \del_{i'<\bar l} \bigg) e^{(l)}_1, 
}
where the primed notation is the same as for \eqref{RK-factor}.
Furthermore, as $\eta$ is a lowest weight vector for $L$, the action of $A^{(k)}_a(v)$ is also upper triangular,
\[
\big[A^{(k)}_a(v)\big]_{ij} \eta = \del_{ij} \prod_{p=1}^\ell\la^{(p)}_{k}(v,c_p)\, \eta  \qq \text{for } i \geq j.
\]
Therefore, taking a product of these matrices, the action of the level-$k$ nested monodromy matrix will also be upper triangular, from which we conclude that $\eta^{(k)}_{+}$ is a lowest weight vector.
The identities (\ref{a1}-\ref{an}) may be found by
\aln{
\Bigg[\prod_{i=1}^{m_{l-1}} R^{(k,l)}_{a a^{l-1}_i}(v,u^{(l-1)}_i) \Bigg]_{j1} (e^{(l)}_1)^{\ot m_{l-1}} &= \del_{j1}\left( \del_{kl}\prod_{i=1}^{m_{k-1}} \frac{q^{-1} v-q\,u^{(k-1)}_i}{v-u^{(k-1)}_i} + \del_{k>l} \right) (e^{(l)}_1)^{\ot m_{l-1}} ,
\\
\Bigg[\prod_{i=m_0}^{1} K^{(k,1)}_{a \tl a^{0}_i}(v,q^{2\theta} u^{(0)}_i) \Bigg]_{j1} (e^{(1)}_1)^{\ot m_0} \hspace{.22cm} &= \del_{j1}\left( \del_{kn}\, \prod_{i=1}^{m_0}\frac{q^{1-\theta}\,v-q^{\theta-1}u^{(0)}_i}{q^{-\theta}v-q^{\theta}u^{(0)}_i} + \del_{k<n} \right) (e^{(1)}_1)^{\ot m_0} ,
\\
\Bigg[\prod_{i=1}^{m_0} K^{(k,1)}_{a a^{0}_i}(v, u^{(0)}_i) \Bigg]_{11} (e^{(1)}_1)^{\ot m_0} \hspace{.22cm} &=  \left( \del_{kn}\, \prod_{i=1}^{m_0}\frac{q\,v-q^{-1} u^{(0)}_i}{v- u^{(0)}_i} + \del_{k<n} \right) (e^{(1)}_1)^{\ot m_0} ,
}
where the matrix elements are taken with respect to the the ``$a$'' space.
Expressions (\ref{d1}-\ref{dn}) are obtained similarly. This concludes the proof in the symplectic case.

The orthogonal case follows by the same arguments and the fact that $\xi_-^{(m_0)}$ is a lowest weight vector with respect to the action of
\[
\mr{K}^{[m_0]} :=  \left(\prod_{i=m_0}^{1} K^{(1,1)}_{a \tl a^{0}_i}(v,q^{2} u^{(0)}_i) \right) \left(\prod_{i=1}^{m_0} K^{(1,1)}_{aa^{0}_i}(v,u^{(0)}_i) \right).
\]
We will prove the latter by induction on $m_0$. The $m_0=1$ case has already been explained above. Let $s\ge1$. We assume that $\xi_-^{(s)}$ is a lowest weight vector for $\mr{K}^{[s]}$ of weight $\lambda^{[s]}_i(v) =1$ for $i<n-1$ and $\lambda^{[s]}_i(v) = \prod_{j=1}^s \frac{q\,v-q^{-1}u}{v-u}$ for $i=n-1,n$.
We write the action of $\mr{K}^{[s+1]}$ on $\xi_-^{(s+1)}$ as
\ali{
\big[\mr{K}^{[s+1]}_{ij}\big]_{ij} \cdot \xi_-^{(s+1)} &= \sum_{b,c=1}^{n} (\mc{E}^{(1)}_{q^{-1},bi}(q^{2}u /v))^{\bar\om}\, \big[\mr{K}^{[s]}\big]_{bc}(\mc{E}^{(1)}_{q^{-1},jc}(u/v))^{\bar\om} \el & \hspace{2cm} \times \big( e^{(1)}_1 \ot \xi_-^{(s)} \ot e^{(1)}_2 - q e^{(1)}_2 \ot \xi_-^{(s)} \ot e^{(1)}_1  \big)
\nn\\[.3em]
&= \sum_{\substack{b,c=1 \\ i\leq b \leq c}}^{n} (\mc{E}^{(1)}_{q^{-1},bi}( q^{2} u /v))^{\bar\om}\, \big[\mr{K}^{[s]}\big]_{bc} (\mc{E}^{(1)}_{q^{-1},jc}(u/v))^{\bar\om}  \cdot e^{(1)}_1  \ot \xi_-^{(s)} \ot e^{(1)}_2
\el
&\qu -q \sum_{\substack{b,c=1 \\ b \leq c \leq j}}^{n} (\mc{E}^{(1)}_{q^{-1},bi}(q^{2}u /v))^{\bar\om}\, \big[\mr{K}^{[s]}\big]_{bc} (\mc{E}^{(1)}_{q^{-1},jc}(u/v))^{\bar\om}  \cdot e^{(1)}_2  \ot \xi_-^{(s)} \ot e^{(1)}_1 \label{K-xi}
}
since $\xi_-^{(s)}$ and $e^{(1)}_1$ are lowest weight vectors in their relevant representations.
Observe that
\aln{
(\mc{E}^{(1)}_{q^{-1},ji}(u/v))^{\bar\om}\, e^{(1)}_1 &= \del_{ij}\bigg(\del_{in} \frac{q\,v-q^{-1}u}{v-u} + \del_{i< n} \bigg) \, e^{(1)}_1 \\ &\qu + \del_{i<n} \del_{jn} \frac{q^{i-j}(q-q^{-1})}{v/u-1}\, e^{(1)}_{n-i+1} 
\intertext{and}
(\mc{E}^{(1)}_{q^{-1},ji}(u/v))^{\bar\om}\, e^{(1)}_2 &= \del_{ij}\bigg(\del_{i,n-1} \frac{q\,v-q^{-1}u}{v-u} + \del_{i\neq n-1} \bigg) \, e^{(1)}_2 - \del_{in} \del_{j,n-1} \frac{q^2-1}{u/v-1} \, e^{(1)}_{1}
\\
& \qu + \del_{i<n-1} \del_{j,n-1} \frac{q^{i-j}(q-q^{-1})}{v/u-1} \, e^{(1)}_{n-i+1} \,.
}
Assume that $i>j$. It clear from above that $\mr{K}^{[s+1]}_{ij} \cdot \xi_-^{(s+1)} = 0$ if $j<n-1$. Hence we only need to consider the case with $i=n$ and $j=n-1$. Then \eqref{K-xi} becomes
\aln{
\big[\mr{K}^{[s+1]}\big]_{n,n-1} \cdot \xi_-^{(s+1)} &= - \bigg( \la_{n}^{[s]}(v)\,\frac{v-u}{q^{-1}v-q\,u}\cdot\frac{q^2-1}{u/v-1} \\ & \qq\qu - q\,\la_{n}^{[s]}(v)\,\frac{q^2-1}{q^2u/v-1}\cdot\frac{q\,v-q^{-1}u}{v-u}\bigg)\, e^{(1)}_1  \ot \xi_-^{(s)} \ot e^{(1)}_1 = 0,
}
as required. Next, assume that $i=j=n$. Then
\aln{
\big[\mr{K}^{[s+1]}\big]_{nn} \cdot \xi_-^{(s+1)} &= \la_{n}^{[s]}(v)\left( \frac{v-u}{q^{-1}v-q\,u} \,e^{(1)}_1 \ot \xi_-^{(s)} \ot e^{(1)}_2  - q\,\frac{q\,v-q^{-1}u}{v-u} \,e^{(1)}_2 \ot \xi_-^{(s)} \ot e^{(1)}_1 \right)
\\ & \qu + q\,\la_{n-1}^{[s]}(v)\,\frac{q^2-1}{q^2 u/v - 1}\cdot\frac{q^{-1}(q-q^{-1})}{v/u-1}\,e^{(1)}_1 \ot \xi_-^{(s)} \ot e^{(1)}_2
\\
& =  \la_{n}^{[s]}(v)\,\frac{q\,v-q^{-1}u}{v-u} \left( e^{(1)}_1 \ot \xi_-^{(s)} \ot e^{(1)}_2  - q\,e^{(1)}_2 \ot \xi_-^{(s)} \ot e^{(1)}_1 \right)
\\
& =  \la_{n}^{[s+1]}(v)\,\xi_-^{(s+1)} .
\intertext{In a similar way, for $i=j=n-1$, we find}
\big[\mr{K}^{[s+1]}\big]_{n-1,n-1} \cdot \xi_-^{(s+1)} &= \la_{n-1}^{[s]}(v)\left( \frac{q\,v-q^{-1}u}{v-u} \,e^{(1)}_1 \ot \xi_-^{(s)} \ot e^{(1)}_2 - q^2\frac{v-u}{v-q^2u} \,e^{(1)}_2 \ot \xi_-^{(s)} \ot e^{(1)}_1 \right)
\\
& \qu - \la_{n}^{[s]}(v)\,\frac{q^{-1}(q-q^{-1})}{q^{-2}v/u - 1}\cdot\frac{q^2-1}{u/v-1} \,e^{(1)}_2 \ot \xi_-^{(s)} \ot e^{(1)}_1
\\
& =  \la_{n-1}^{[s]}(v)\,\frac{q\,v-q^{-1}u}{v-u} \left( e^{(1)}_1 \ot \xi_-^{(s)} \ot e^{(1)}_2  - q\,e^{(1)}_2 \ot \xi_-^{(s)} \ot e^{(1)}_1 \right)
\\
& =  \la_{n-1}^{[s+1]}(v)\,\xi_-^{(s+1)} .
}
Lastly, when $i=j<n-1$ we obtain $\big[\mr{K}^{[s+1]}\big]_{ii} \cdot \xi_-^{(s+1)} = \xi_-^{(s+1)}$.
Then, using the arguments similar to those in the symplectic case yields the wanted result.
\end{proof}

\subsection{Transfer matrices, Bethe vectors and Bethe equations}

Recall the notion of level-$k$ monodromy matrices, viz.\ Definition~\ref{D:mono}.

%%% level-k transfer matrices

\begin{defn}
We define level-$0$ transfer matrix by
\equ{
\tau(v) := \tr_a T_a(v) = \tr_a A^{(1)}_a(v) + \tr_a D^{(1)}_a(v) .
}
For all $1\le k \le n-1$ we define level-$k$ transfer matrices by
\[
\tau^{(k)}(v;\bm u^{(0\dots k-1)}) := \tr_a A^{(k)}_{a}(v;\bm u^{(0\dots k-1)}) 
\]
and
\[
\wt\tau^{(k)}(v;\bm u^{(0\dots k-1)}) := \tr_a D^{(k)}_{a}(v;\bm u^{(0\dots k-1)}) .
\]
\end{defn}

%%% level-k creation operators

Next, for each level of nesting, $0 \le k \le n-1$, we introduce $m_k$-magnon creation operators.

\begin{defn}
We define level-$0$ creation operator by
\equ{
\mr{B}^{(0)}_{\bm a^{\tl0,0}}(\bm u^{(0)}) := \prod_{i=1}^{m_0} \be_{\ta^0_i a^0_i} (u^{(0)}_i) 
}
where
\equ{
\be_{\ta^0_i a^0_i} (u^{(0)}_{i}) := \sum_{j,k=1}^n q^{-j}\, t_{\bar\jmath, n+k}(u^{(0)}_{i}) \ot e^{(1)*}_k \ot e^{(1)*}_j  \in  \End(L^{(0)}) \ot V^{(1)*}_{\ta^0_i} \ot  V^{(1)*}_{a^0_i} . \label{beta}
}
For all $1\le k \le n-1$ we define level-$k$ creation operators by
\ali{
\mr{B}^{(k)}_{\bm a^{k}}(\bm u^{(k)}; \bm u^{(0\dots k-1)}) &:= \prod_{i=1}^{m_k} B^{(k+1)}_{a^{k}_i}(u^{(k)}_i;\bm u^{(0\dots k-1)}) 
}
where 
\equ{
B^{(k+1)}_{a^{k}_i}(u^{(k)}_i;\bm u^{(0\dots k-1)}) := \sum_{j=1}^{n-k} \big[ A^{(k)}_{a^{k}_i}(u^{(k)}_i;\bm u^{(0\dots k-1)}) \big]_{1,1+j} \ot e^{(k+1)*}_j \in \End(L^{(k)}) \ot  V^{(k+1)*}_{a^k_i}. \label{beta1}
}
\end{defn}

%%% level-k Bethe vectors

\nc{\BV}[2]{\Phi^{(#1)}_\theta (\bm u^{(#1\dots n-1)};\bm u^{(0\dots#2)})}

\nc{\BVi}{\Phi^{(1)}_\theta(\bm u^{(1\dots n-1)};\bm u^{(0)})}
\nc{\BVii}{\Phi^{(2)}_\theta(\bm u^{(2\dots n-1)};\bm u^{(0,1)})}

Recall the notion of nested vacuum vector $\eta^{(k)}_{\pm}$, viz.\ \eqref{eta-k}. The level-$(n\!-\!1)$ nested vacuum vector is our reference state for constructing the (off-shell) Bethe vectors.

\begin{defn} \label{D:BV}
For all $1\le k \le n-1$ we define level-$k$ Bethe vectors by
\equ{
\BV{k}{k-1} := \prod_{i=k}^{n-1} \mr{B}^{(i)}_{\bm a^{i}}(\bm u^{(i)}; \bm u^{(0\dots i-1)}) \cdot \eta^{(n-1)}_{-\theta}. \label{BV-k}
}
The level-$0$ Bethe vector is defined by
\equ{
\Phi^{(0)}_\theta(\bm u^{(0\dots n-1)}) := \mr{B}^{(0)}_{\bm a^{\tl 0,0}}(\bm u^{(0)}) \, \BVi . \label{BV-0}
}
\end{defn}

Note that vector $\BV{k}{k-1}$ is an element of the level-$k$ quantum space $L^{(k)}$ and has $\bm u^{(0\dots k-1)}$ and $\bm c$ as its free parameters.

For $0\le k \le n-1$ set $\mf{S}_{\bm m_{k\dots n-1}} := \mf{S}_{m_k}\times \cdots \times \mf{S}_{m_{n-1}}$.
For any $\si^{(l)} \in \mf{S}_{m_l}$ with $k\le l \le n-1$ define an action of $\mf{S}_{\bm m_{k\dots n-1}}$ on $\BV{k}{k-1}$ by 
\[
\si^{(l)}: \bm u^{(k\dots n-1)} \mapsto \bm u^{(k\dots n-1)}_{\si^{(l)}} := ( \bm u^{(k)} , \dots , \bm u^{(l)}_{\si^{(l)}}, \dots , \bm u^{(n-1)} \} , \qu
\bm u^{(l)}_{\si^{(l)}} := (u^{(l)}_{\si^{(l)}(1)}, \dots , u^{(l)}_{\si^{(l)}(m_l)}).
\]
For further convenience we set $\si^{(l)}_j\in\mf{S}_{m_l}$ to be the $j$-cycle such that 
\[
\bm u^{(l)}_{\si^{(l)}_j} = (u^{(l)}_j,u^{(l)}_{j+1},\dots,u^{(l)}_{m_l},u^{(l)}_1,\dots,u^{(l)}_{j-1}) .
\]
We will also make use of the notation
\[
\bm u^{(l)}_{\si^{(l)}_j\!,\,u^{(l)}_j\to v} := (v,u^{(l)}_{j+1},\dots,u^{(l)}_{m_l},u^{(l)}_1,\dots,u^{(l)}_{j-1}) .
\]

\begin{lemma} \label{L:BV-symm}
Bethe vector $\BV{k}{k-1}$ is invariant under the action of $\mf{S}_{\bm m_{k\dots n-1}}$.
\end{lemma}

\begin{proof} 
This follows using standard arguments, the fact that $\check{R}$-matrices act on $\eta^{(n-1)}_{-\theta}$ by $1$, and relations
\aln{
\be_{\ta^0_i a^0_i}(u^{(0)}_i)\,\be_{\ta^0_{i+1} a^0_{i+1}}(u^{(0)}_{i+1}) = \be_{\ta^0_i a^0_i}(u^{(0)}_{i+1})\,\be_{\ta^0_{i+1} a^0_{i+1}}(u^{(0)}_{i}) \,\big( \check{R}^{(1,1)}_{a^0_ia^0_{i+1}}(u^{(0)}_i,u^{(0)}_{i+1}) \big)^{-1} \check{R}^{(1,1)}_{\ta^0_i\ta^0_{i+1}}(u^{(0)}_i,u^{(0)}_{i+1}) 
}
and
\aln{
& B^{(k+1)}_{a^{k}_i}(u^{(k)}_i;\bm u^{(0\dots k-1)})  B^{(k+1)}_{a^{k}_{i+1}}(u^{(k)}_{i+1};\bm u^{(0\dots k-1)}) \\ 
& \qq \equiv
B^{(k+1)}_{a^{k}_i}(u^{(k)}_{i+1};\bm u^{(0\dots k-1)})
B^{(k+1)}_{a^{k}_{i+1}}(u^{(k)}_i;\bm u^{(0\dots k-1)}) \check{R}^{(k+1,k+1)}_{a^k_i a^k_{i+1}} (u^{(k)}_i,u^{(k)}_{i+1})
}
which follow from \eqref{BB} and \eqref{gln:bb}; here $\equiv$ denotes equality of operators in the space $L^{(k)}$.
\end{proof}

%%% gln eigenvectors

\begin{thrm} \label{T:gln}
Bethe vector $\Phi^{(1)}_{\theta}(\bm u^{(1\dots n-1)};\bm u^{(0)})$ is an eigenvector of $\tau^{(1)}(v;\bm u^{(0)})$ with the eigenvalue
\ali{
\La^{(1)}(v;\bm u^{(1\dots n-1)}; \bm u^{(0)}) &:= \prod_{i=1}^{m_1} \frac{q\,v-q^{-1}u^{(1)}_i}{v-u^{(1)}_i} \prod_{i=1}^\ell\la^{(i)}_{1}(v) \el
& \qu\; + \sum_{k=2}^{n-2} \prod_{i=1}^{m_{k-1}} \frac{q^{-1} v-q\,u^{(k-1)}_i}{v-u^{(k-1)}_i} \prod_{i=1}^{m_k} \frac{q\,v-q^{-1}u^{(k)}_i}{v-u^{(k)}_i} \prod_{i=1}^\ell\la^{(i)}_{k}(v)  \el 
& \qu\; + \prod_{i=1}^{m_{n-2}} \frac{q^{-1} v-q\,u^{(n-2)}_i}{v-u^{(n-2)}_i} \prod_{i=1}^{m_{n-1}} \frac{q\, v-q^{-1}\,u^{(n-1)}_i}{v-u^{(n-1)}_i} \prod_{i=1}^{m_{0}} \frac{q^{\theta'}\,v-q^{-\theta'}u^{(0)}_i}{v-u^{(0)}_i} \prod_{i=1}^\ell\la^{(i)}_{n-1}(v) \el
& \qu\; + \prod_{i=1}^{m_0} \frac{q^{2-\theta'}v-q^{-2+\theta'} u^{(0)}_i}{v-u^{(0)}_i} \prod_{i=1}^{m_{n-1}} \frac{q^{-1} v-q\,u^{(n-1)}_i}{v-u^{(n-1)}_i} \prod_{i=1}^\ell\la^{(i)}_{n}(v) \label{LA} 
\intertext{and an eigenvector of $\wt\tau^{(1)}(v;\bm u^{(0)})$ with the eigenvalue}
\wt{\La}^{(1)}(v;\bm u^{(1\dots n-1)};\bm u^{(0)}) & := \prod_{i=1}^{m_1} \frac{v-q^{2\ka}u^{(1)}_i}{q\,v-q^{2\ka-1} u^{(1)}_i} \prod_{i=1}^\ell\la^{(i)}_{2n}(v) \el
& \qu\; + \sum_{k=2}^{n-2} \prod_{i=1}^{m_{k-1}} \frac{q\,v-q^{2\ka-2k+1} u^{(k-1)}_i}{v-q^{2\ka-2k+2}u^{(k-1)}_i} \prod_{i=1}^{m_k} \frac{v-q^{2\ka-2k+2}u^{(k)}_i}{q\,v-q^{2\ka-2k+1}u^{(k)}_i} \prod_{i=1}^\ell\la^{(i)}_{2n-k+1}(v)  \el 
& \qu \; + \prod_{i=1}^{m_{n-2}} \frac{q\,v-q^{3-2\theta} u^{(n-2)}_i}{v-q^{4-2\theta}u^{(n-2)}_i} \prod_{i=1}^{m_{n-1}} \frac{v-q^{4-2\theta}u^{(n-1)}_i}{q\,v-q^{3-2\theta}u^{(n-1)}_i} \prod_{i=1}^{m_{0}} \frac{q^{-\theta'}\,v-q^{\theta'}u^{(0)}_i}{v-u^{(0)}_i} 
\prod_{i=1}^\ell\la^{(i)}_{n+2}(v)  \el 
& \qu\; + \prod_{i=1}^{m_{n-1}} \frac{q\,v-q^{1-2\theta} u^{(n-1)}_i}{v-q^{2-2\theta}u^{(n-1)}_i} \prod_{i=1}^{m_0} \frac{q^{-2+\theta'}\,v-q^{2-\theta'}u^{(0)}_i}{v-u^{(0)}_i} \prod_{i=1}^\ell\la^{(i)}_{n+1}(v) \label{tLA} 
}
provided 
\equ{
\Res{v\to u^{(k)}_j} \La^{(1)}(v;\bm u^{(1\dots n-1)};\bm u^{(0)}) = 0 \qu\text{for}\;\; 1\le k \le n-1,\; 1\le j \le m_k. \label{BE-gln}
}
\end{thrm}

\begin{proof}
{\it Step 1. $\BVi$ is an eigenvector of $\tau^{(1)}(v,\bm u^{(0)})$.}
This is a standard result (see e.g.\ \cite{BeRa08}), hence we give an outline of the proof only. We rewrite the exchange relations \eqref{gln:ab} and \eqref{gln:db} as
\aln{
\mr{a}^{(k)}(v)\, B^{(k+1)}_{a^k_i}(u^{(k)}_i) &= \frac{q\,v-q^{-1}u^{(k)}_i}{v-u^{(k)}_i}\,B^{(k+1)}_{a^k_i}(u^{(k)}_i) \,\mr{a}^{(k)}(v) \\ &\qu - \frac{v/u^{(k)}_i}{v-u^{(k)}_i}\Res{w\to u^{(k)}_i}\frac{q\,w-q^{-1}u^{(k)}_i}{w-u^{(k)}_i}\,B^{(k+1)}_{a^k_i}(v) \, \mr{a}^{(k)}(w) , \\[.75em]
A^{(k+1)}_a(v)\,B^{(k+1)}_{a^k_i}(u^{(k)}_i) &= B^{(k+1)}_{a^k_i}(u^{(k)}_i) \,A^{(k+1)}_a(v)\,R^{(k+1)}_{aa^k_i}(v,u^{(k)}_i) \\ &\qu - \frac{v/u^{(k)}_i}{v-u^{(k)}_i}\Res{w\to u^{(k)}_i} B^{(k+1)}_{a^k_i}(v) A^{(k+1)}_a(w)\,R^{(k+1,k+1)}_{aa^k_i}(w,u^{(k)}_i).
}
Then, using the usual symmetry arguments for the Bethe vector, viz.\ Lemma \ref{L:BV-symm}, we obtain
\aln{
& \tau^{(1)}(v;\bm u^{(0)})\,\BVi \\ 
&\qu = \mr{B}^{(1)}_{\bm a^{1}}(\bm u^{(1)};\bm u^{(0)}) \Bigg( \prod_{i=1}^{m_1} \frac{q\,v-q^{-1}u^{(1)}_i}{v-u^{(1)}_i} \, \mr{a}^{(1)}(v;\bm u^{(0)}) + \tau^{(2)}(v;\bm u^{(0,1)}) \Bigg) \BVii \\ 
& \qq - \sum_{j=1}^{m_1} \frac{v/u^{(1)}_j}{v-u^{(1)}_j} \, \mr{B}^{(1)}_{\bm a^{1}}(\bm u^{(1)}_{\si^{(1)}_j,u^{(1)}_j\to v};\bm u^{(0)}) \\ 
& \qq\qu \times \Res{w\to u^{(1)}_j}\!\Bigg( \prod_{i=1}^{m_1} \frac{q\,w-q^{-1}u^{(1)}_i}{w-u^{(1)}_i} \,\mr{a}^{(1)}(w;\bm u^{(0)}) + \tau^{(2)}(w;\bm u^{(0,1)}_{\si^{(1)}_j,u^{(1)}_j\to v}) \Bigg) \Phi_{\theta}^{(2)}(\bm u^{(2\dots n-1)};\bm u^{(0,1)}_{\si^{(1)}_j}) .
}
Proceeding in the same way and using Lemma \ref{L:a-action} we find
\aln{
\tau^{(1)}(v;\bm u^{(0)})\,\BVi = \La^{(1)}(v;\bm u^{(1\dots n-1)};\bm u^{(0)})\,\BVi
}
provided \eqref{BE-gln} holds, as required.

{\noindent\it Step 2. $\BVi$ is an eigenvector of $\wt\tau^{(1)}(v,\bm u^{(0)})$.} 
It follows from Lemma \ref{L:a,d-rll} that transfer matrices $\tau^{(1)}(v;\bm u^{(0)})$ and $\wt\tau^{(1)}(v;\bm u^{(0)})$ form a family of commutative operators in the space $L^{(1)}$. They can thus be diagonalized simultaneously. Assuming \eqref{BE-gln} holds, it is sufficient to focus on the wanted terms in the exchange relations. In particular, it follows from \eqref{DRA=ARD} that
\aln{
\mr{d}^{(k)}(v)\, B^{(k+1)}_{a^k_i}(u^{(k)}_i) &= \frac{v-q^{2\ka-2k+2}u^{(k)}_i}{q\,v-q^{2\ka-2+1}u^{(k)}_i}\,B^{(k+1)}_{a^k_i}(u^{(k)}_i) \,\mr{d}^{(k)}(v) + UWT, \\[.75em]
D^{(k+1)}_a(v)\,B^{(k+1)}_{a^k_i}(u^{(k)}_i) &= B^{(k+1)}_{a^k_i}(u^{(k)}_i)\,R^{(k+1)}_{q^{-1},aa^k_i}(v,q^{2\ka-2k+2}u^{(k)}_i) \,D^{(k+1)}_a(v) + UWT ,
}
where $UWT$ denotes the unwanted terms. The eigenvalue \eqref{tLA} now follows by Lemma \ref{L:a-action} and the standard arguments.
\end{proof}

%%% First main result

The Theorem below is our first main result.

\begin{thrm} \label{T:sosp} 
Bethe vector $\Phi^{(0)}_{\theta}(\bm u^{(0\dots n-1)})$ is an eigenvector of $\tau(v)$ with eigenvalue
\equ{
\La(v;\bm u^{(0\dots n-1)}) := \La^{(1)}(v;\bm u^{(1\dots n-1)};\bm u^{(0)}) + \wt{\La}^{(1)}(v;\bm u^{(1\dots n-1)};\bm u^{(0)}) \label{GA}
}
provided 
\ali{
\Res{v\to u^{(0)}_j} \La(v;\bm u^{(0\dots n-1)}) &= 0 \qu\text{for}\;\; 1\le j \le m_0,\; \label{BE-top}
\;\;\text{and}  \\
\Res{v\to u^{(k)}_j} \La^{(1)}\big(v;\bm u^{(1\dots n-1)};\bm u^{(0)}\big) &= 0 \qu\text{for}\;\; 1\le j \le m_k, \; 1 \le k \le n-1 .   \label{BE-bottom}
}
\end{thrm}

\begin{proof}
Using \eqref{RQ} and \eqref{RUK1}--\eqref{KU=KP}, we deduce $R^{(1,1)}_{21}(u,v) = R^{(1,1)}_{q^{-1},12}(v,u)$ and
\[
(K^{(1,1)}_{12}(u,q^{2\ka }v))^{-1} = K^{(1,1)}_{21}(v,q^{2\theta} u) = K^{(1,1)}_{q^{-1},12}(u,q^{-2\theta} v) .
\]
These relations allow us to rewrite \eqref{AB} and \eqref{DB} as
\aln{
A^{\pm}_1(v) B^{\pm}_2(u) &= R^{(1,1)}_{q^{-1},12}(v,u) B^{\pm}_2(u) A^{\pm}_1(v) K^{(1,1)}_{12}(v,q^{2\theta} u) \\ & \qu - \frac1{v-u}\Res{w\to u} R^{(1,1)}_{q^{-1},12}(w,u) B^{\pm}_1(v) A^{\pm}_2(u) K^{(1,1)}_{12}(w,q^{2\theta} u) , 
\\[.5em]
D^{\pm}_1(v) B^{\pm}_2(u) &= K^{(1,1)}_{q^{-1},12}(v,q^{-2\theta} u) B^{\pm}_2(u) D^{\pm}_1(v) R^{(1,1)}_{12}(v,u) \\ & \qu - \frac1{v-u}\Res{w\to u} K^{(1,1)}_{q^{-1},12}(w,q^{-2\theta} u) B^{\pm}_1(v) D^{\pm}_2(u) R^{(1,1)}_{12}(w,u) .
}
Then, using \eqref{RUK1}, \eqref{beta} and replacing the $A$ and $D$ operators with their images in $\End(L^{(0)})$, we obtain
\aln{
A^{(1)}_a(v)\, \be_{\ta^0_i a^0_i} (u^{(0)}_{i}) &= \be_{\ta^0_i a^0_i} (u^{(0)}_{i}) \, K^{(1,1)}_{aa^0_i}(v,u^{(0)}_{i}) A^{(1)}_a(v) K^{(1,1)}_{a\ta^0_i}(v,q^{2\theta} u^{(0)}_{i}) \\ & \qu - \frac1{v-u^{(0)}_i} \,\be_{\ta^0_i a^0_i} (v)  \Res{w\to u^{(0)}_{i}} K^{(1,1)}_{aa^0_i}(w,u^{(0)}_{i}) A^{(1)}_a(w) K^{(1,1)}_{a\ta^0_i}(w,q^{2\theta} u^{(0)}_{i}) , 
\\[.5em]
D^{(1)}_a(v)\, \be_{\ta^0_i a^0_i} (u^{(0)}_{i}) &= \be_{\ta^0_i a^0_i} (u^{(0)}_{i}) \, R^{(1,1)}_{aa^0_i}(v,q^{-2\theta} u^{(0)}_{i}) D^{(1)}_a(v) R^{(1,1)}_{a\ta^0_i}(v,u^{(0)}_{i}) \\ & \qu - \frac1{v-u^{(0)}_i} \, \be_{\ta^0_i a^0_i} (v) \Res{w\to u^{(0)}_{i}} R^{(1,1)}_{aa^0_i}(w,q^{-2\theta} u^{(0)}_{i}) D^{(1)}_a(v) R^{(1,1)}_{a\ta^0_i}(w,u^{(0)}_{i}) .
}
The relations above together with Lemma \ref{L:BV-symm} and the standard symmetry arguments imply that
\aln{
& \tau(v)\, \Phi^{(0)}_{\theta}(\bm u^{(0\dots n-1)}) \\ 
& \qu = \mr{B}^{(0)}_{\bm a^{\tl 0,0}}(\bm u^{(0)}) \, \Big( \tau^{(1)}(v;\bm u^{(0)}) + \wt{\tau}^{(1)}(v;\bm u^{(0)}) \Big) \,\Phi^{(1)}_{\theta}(\bm u^{(1\dots n-1)};\bm u^{(0)})\\
& \qq - \sum_{j=1}^{m_0} \frac1{v-u^{(0)}_j} \, \mr{B}^{(0)}_{\bm a^{\tl 0,0}}(\bm u^{(0)}_{\si^{(0)}_j,u^{(0)}_j\to v}) \Res{w\to u^{(0)}_j} \Big( \tau^{(1)}(w;\bm u^{(0)}_{\si^{(0)}_j}) + \wt{\tau}^{(1)}(w;\bm u^{(0)}_{\si^{(0)}_j}) \Big)\, \Phi^{(1)}_{\theta}(\bm u^{(1\dots n-1)};\bm u^{(0)}_{\si^{(0)}_j}) .
}
Theorem \eqref{T:gln} allows us to replace $ \tau^{(1)}(w;\bm u^{(0)}_{\si^{(0)}_j})$ and $\wt{\tau}^{(1)}(w;\bm u^{(0)}_{\si^{(0)}_j})$ with their eigenvalues, provided \eqref{BE-bottom} holds, giving
\aln{
& \tau(v)\, \Phi^{(0)}_{\theta}(\bm u^{(0\dots n-1)}) \\ 
& \qu = \mr{B}^{(0)}_{\bm a^{\tl 0,0}}(\bm u^{(0)}) \,  \Lambda(v;\bm u^{(0...n-1)})  \,\Phi^{(1)}_{\theta}(\bm u^{(1\dots n-1)};\bm u^{(0)})\\
& \qq - \sum_{j=1}^{m_0} \frac1{v-u^{(0)}_j} \, \mr{B}^{(0)}_{\bm a^{\tl 0,0}}(\bm u^{(0)}_{\si^{(0)}_j,u^{(0)}_j\to v}) \Res{w\to u^{(0)}_j}  \Lambda(w;\bm u^{(0)}_{\si^{(0)}_j},\bm u^{(1...n-1)}) \, \Phi^{(1)}_{\theta}(\bm u^{(1\dots n-1)};\bm u^{(0)}_{\si^{(0)}_j}) .
}
Noting that, from its definition and the exact forms of $\Lambda^{(1)}$ and $\wt{\Lambda}^{(1)}$, we have $\Lambda(w;\bm u^{(0)}_{\si^{(0)}_j},\bm u^{(1...n-1)}) = \Lambda(w;\bm u^{(0...n-1)})$.
Therefore, \eqref{BE-top} implies that each term in the sum on the second line individually vanishes, and we are left with the desired result.
\end{proof}

\begin{rmk}[i] \label{R:BEs}
The equations \eqref{BE-bottom} are Bethe equations for a $U_q(\mfgl_{n})$-symmetric spin chain, with an additional factor when $j=n-1$ due to level-0 excitations. 
For convenience, set $u^{(n)}_j := u^{(0)}_j$ and $m_n := m_0$. 
Then the explicit form of the equations \eqref{BE-bottom} and \eqref{BE-top} in the symplectic case is 
\ali{
\prod_{i=1}^\ell\frac{\la^{(i)}_{1}(u_j^{(1)})}{\la^{(i)}_{2}(u_j^{(1)})} &= \prod_{\substack{i=1\\ i\ne j}}^{m_1} \frac{q^{-1} u_j^{(1)}-q\,u^{(1)}_i}{q\,u_j^{(1)}-q^{-1}u^{(1)}_i} \prod_{i=1}^{m_{2}} \frac{q\,u_j^{(1)}-q^{-1}u^{(2)}_i}{u_j^{(1)}-u^{(2)}_i} \,, 
\label{BE-sp-1}
\\
\prod_{i=1}^\ell \frac{\la^{(i)}_{k}(u^{(k)}_j)}{\la^{(i)}_{k+1}(u^{(k)}_j)} &= \prod_{i=1}^{m_{k-1}} \frac{u^{(k)}_j-u^{(k-1)}_i}{q^{-1} u^{(k)}_j-q\,u^{(k-1)}_i}  \prod_{\substack{i=1\\ i\ne j}}^{m_{k}} \frac{q^{-1} u^{(k)}_j-q\,u^{(k)}_i}{q\,u^{(k)}_j-q^{-1}u^{(k)}_i} \prod_{i=1}^{m_{k+1}} \frac{q\,u^{(k)}_j-q^{-1}u^{(k+1)}_i}{u^{(k)}_j-u^{(k+1)}_i} \,,
\label{BE-sp-k}
\\
\prod_{i=1}^\ell \frac{\la^{(i)}_{n-1}(u^{(n-1)}_j)}{\la^{(i)}_{n}(u^{(n-1)}_j)} &= \prod_{i=1}^{m_{n-2}} \frac{u^{(n-1)}_j-u^{(n-2)}_i}{q^{-1} u^{(n-1)}_j-q\,u^{(n-2)}_i} \prod_{\substack{i=1\\ i\ne j}}^{m_{n-1}} \frac{q^{-1} u^{(n-1)}_j-q\,u^{(n-1)}_i}{q\,u^{(n-1)}_j-q^{-1}u^{(n-1)}_i} \prod_{i=1}^{m_n} \frac{q^{2}\,u^{(n-1)}_j-q^{-2}u^{(n)}_i}{u^{(n-1)}_j-u^{(n)}_i} \,,
\label{BE-sp-n-1}
\\
\prod_{i=1}^\ell \frac{\la^{(i)}_{n}(u^{(n)}_j)}{\la^{(i)}_{n+1}(u^{(n)}_j)} &=  \prod_{i=1}^{m_{n-1}} \frac{u^{(n)}_j-u^{(n-1)}_i}{q^{-2}u^{(n)}_j-q^{2} u^{(n-1)}_i} \prod_{\substack{i=1\\i\ne j}}^{m_n} \frac{q^{-2}\,u^{(n)}_j-q^{2}u^{(n)}_i}{q^{2}\,u^{(n)}_j-q^{-2}u^{(n)}_i} 
\label{BE-sp-n}
}
for $2\le k \le n-2$ and all allowed $j$, and weights given by \eqref{la-skew}.

\noindent(ii). In the orthogonal case, the Bethe equations for $k=1,\dots, n-3$ are identical to the symplectic case.
For $k=n-2$ and $n-1$, the equations are replaced by, respectively,
\ali{
\prod_{i=1}^\ell \frac{\la^{(i)}_{n-2}(u^{(n-2)}_j)}{\la^{(i)}_{n-1}(u^{(n-2)}_j)} &= \prod_{i=1}^{m_{n-3}} \frac{u^{(n-2)}_j-u^{(n-3)}_i}{q^{-1} u^{(n-2)}_j-q\,u^{(n-3)}_i} \prod_{\substack{i=1\\i\ne j}}^{m_{n-2}} \frac{q^{-1} u^{(n-2)}_j-q\,u^{(n-2)}_i}{q\,u^{(n-2)}_j-q^{-1}u^{(n-2)}_i} \el & \qu \times \prod_{i=1}^{m_{n-1}} \frac{q\,u^{(n-2)}_j-q^{-1}u^{(n-1)}_i}{u^{(n-2)}_j-u^{(n-1)}_i} \prod_{i=1}^{m_n} \frac{q\,u^{(n-2)}_j-q^{-1}u^{(n)}_i}{u^{(n-2)}_j-u^{(n)}_i} \,, \label{BE-so-n-2}
\\[0.75em]
\prod_{i=1}^\ell \frac{\la^{(i)}_{n-1}(u^{(n-1)}_j)}{\la^{(i)}_{n}(u^{(n-1)}_j) } &= \prod_{i=1}^{m_{n-2}} \frac{u^{(n-1)}_j-u^{(n-2)}_i}{q^{-1} u^{(n-1)}_j-q\,u^{(n-2)}_i} \prod_{\substack{i=1\\i\ne j}}^{m_{n-1}} \frac{q^{-1} u^{(n-1)}_j-q\,u^{(n-1)}_i}{q\,u^{(n-1)}_j-q^{-1}u^{(n-1)}_i}  \,, \label{BE-so-n-1}
}
for all allowed $j$.
For the level-0 Bethe equations, however, the eigenvalue contains four poles at each Bethe root, rather than two. 
Through use of the identity (see Remark 6.6 in \cite{GRW20}),
\[
\frac{\lambda_{n+2}(u)}{\lambda_{n}(u)} = \frac{\lambda_{n+1}(u)}{\lambda_{n-1}(u)},
\]
the resulting expression may be factorised to give the following Bethe equations, for $1\leq j \leq m_n$,
\ali{ \label{BE-so-factor}
&\Bigg( \prod_{i=1}^{m_{n-2}} \frac{q^{-1} u^{(n)}_j-q\,u^{(n-2)}_i}{u^{(n)}_j-u^{(n-2)}_i} \prod_{i=1}^{m_{n-1}} \frac{q\,u^{(n)}_j-q^{-1}u^{(n-1)}_i}{q^{-1}\,u^{(n)}_j-q\,u^{(n-1)}_i} + \prod_{i=1}^{\ell}\frac{ \la^{(i)}_{n}(u^{(n)}_j)}{\la^{(i)}_{n-1}(u^{(n)}_j)}   \Bigg)
\el
&\hspace{3cm} \times \Bigg( \prod_{\substack{i=1 \\ i \neq j}}^{m_n} \frac{q\,u^{(n)}_j-q^{-1}u^{(n)}_i}{q^{-1}\,u^{(n)}_j-q\,u^{(n)}_i}\prod_{i=1}^{m_{n-2}} \frac{q^{-1} u^{(n)}_j-q\,u^{(n-2)}_i}{u^{(n)}_j-u^{(n-2)}_i}  - \prod_{i=1}^{\ell} \frac{ \la^{(i)}_{n+1}(u^{(n)}_j)}{\la^{(i)}_{n-1}(u^{(n)}_j)} \Bigg)  = 0.
}
Observe that setting the first factor to zero is exactly equivalent to \eqref{BE-so-n-1}, the level-$(n-1)$ set of equations, noting that the sign discrepancy is due to the product in \eqref{BE-so-n-1} excluding the $i=j$ index.
This factorisation is due to an automorphism of the Dynkin diagram of type $D_n$, which exchanges the two branching nodes.
This symmetry of Dynkin diagram is broken by our nesting procedure, and so we obtained only a single set of Bethe equations for the level-$(n-1)$ Bethe roots. 
Extending this to the level-0, we set the right-hand factor above to zero to give the level-0 Bethe equations
\ali{
 \prod_{i=1}^\ell \frac{\la^{(i)}_{n-1}(u^{(n)}_j,c_i)}{\la^{(i)}_{n+1}(u^{(n)}_j,c_i)} &=  \prod_{i=1}^{m_{n-2}} \frac{u^{(n)}_j-u^{(n-2)}_i}{q^{-1}u^{(n)}_j-q\,u^{(n-2)}_i}  \prod_{\substack{i=1\\i\ne j}}^{m_n} \frac{q^{-1}u^{(n)}_j-q\,u^{(n)}_i}{q\,u^{(n)}_j-q^{-1}u^{(n)}_i} \label{BE-so-n}
}
for all allowed $j$.

\noindent(iii). Rather than taking \eqref{BE-top} and \eqref{BE-bottom} separately, we could instead attempt to recover the Bethe equations from the condition $\Res{v \to u^{(k)}_j} \Lambda(v;\bm u^{(0\dots n-1)}) = 0$ for all Bethe roots $u^{(k)}_j$, $0 \leq k \leq n-1$. 
As one might expect, this turns out to be directly equivalent to \eqref{BE-bottom} for $1 \leq k \leq n-2$, however, in the $k=n-1$ case we obtain a factorisation identical to \eqref{BE-so-factor},
\aln{
&\Bigg( \prod_{i=1}^{m_{n-2}} \frac{q^{-1} u^{(n-1)}_j-q\,u^{(n-2)}_i}{u^{(n-1)}_j-u^{(n-2)}_i} \prod_{\substack{i=1 \\ i \neq j}}^{m_{n-1}} \frac{q\,u^{(n-1)}_j-q^{-1}u^{(n-1)}_i}{q^{-1}\,u^{(n-1)}_j-q\,u^{(n-1)}_i} - \prod_{i=1}^{\ell}\frac{ \la^{(i)}_{n}(u^{(n-1)}_j)}{\la^{(i)}_{n-1}(u^{(n-1)}_j)} \Bigg)
\el
&\hspace{3cm} \times \Bigg( \prod_{i=1}^{m_n} \frac{q\,u^{(n-1)}_j-q^{-1}u^{(n)}_i}{q^{-1}\,u^{(n-1)}_j-q\,u^{(n)}_i}\prod_{i=1}^{m_{n-2}} \frac{q^{-1} u^{(n-1)}_j-q\,u^{(n-2)}_i}{u^{(n-1)}_j-u^{(n-2)}_i}  + \prod_{i=1}^{\ell} \frac{ \la^{(i)}_{n+1}(u^{(n-1)}_j)}{\la^{(i)}_{n-1}(u^{(n-1)}_j)} \Bigg)  = 0,
}
for $1 \leq j \leq m_{n-1}$.
This again reflects the fact that the symmetry of the Dynkin diagram of type $D_n$ is unbroken at level-0. One might expect that a nesting procedure of the type employed in \cite{MaRa97} for rational closed spin chains and in \cite{Gom18} for rational open spin chains, in which the chain of symmetry subalgebras is $D_n \supset D_{n-1} \supset \dots \supset D_1$, would preserve this symmetry of the Dynkin diagram at all levels of nesting.
\end{rmk}

\begin{rmk}
For $n=2$, the Bethe equations \eqref{BE-so-n-1} and \eqref{BE-so-n} decouple into two sets of Bethe equations for $U_q(\mf{sl}_2)$-symmetric spin chains, and can be solved separately. This is consistent with the isomorphism $\mf{so}_{4} \cong \mf{sl}_2 \oplus \mf{sl}_2$.
Similarly, for $n=3$, the isomorphism  $\mf{so}_{6} \cong \mf{sl}_4$ is borne out in the Bethe equations \eqref{BE-so-n-2}, \eqref{BE-so-n-1} and \eqref{BE-so-n}. 
\end{rmk}

\begin{rmk} \label{R:Bethe-Dynkin}
Let $a_{ij}$ denote the matrix entries of a connected Dynkin diagram of type $C_n$ or $D_n$ and let $I$ denote the set of its nodes. 
Then put $d_1 = \ldots = d_{n} = 1$ except $d_n=2$ for $C_n$. 
Upon substituting $u^{(k)}_j \to q^{\wt d_k} z^{(k)}_j$, where $\wt d_k = \sum_{i=1}^k d_i$ except $\wt d_n = \sum_{i=1}^{n-1} d_i$ for $D_n$, and taking into account \eqref{la-symm} and \eqref{la-skew}, Bethe equations above can be written as
\[
\prod_{i=1}^\ell \frac{\la_k(q^{\wt d_k} z^{(k)}_j)}{\la_{k+1}(q^{\wt d_k} z^{(k)}_j)} = - \prod_{l \in I} \prod_{i=1}^{m_l} \frac{z^{(k)}_j - q^{d_k a_{kl}} z^{(l)}_i}{q^{d_k a_{kl}} z^{(k)}_j- z^{(l)}_i} \,,
\]
for $1\le k \le n$ and all allowed $j$.
\end{rmk}

%%%%%%%%%%%%%%%%%%%%%%%%%%%%%%%%%%%%%%%%%%%%%%%%%%%%%%%%%%%%%%%%%%%%
%
%%%%%%%%%%%%%%%%%%%%%%%%%%%%%%%%%%%%%%%%%%%%%%%%%%%%%%%%%%%%%%%%%%%%

\subsection{A nearest-neighbour interaction Hamiltonian}

In the case where $L(\la^{(i)})_{c_i} \cong \C^n$, so that $L \cong (\C^n)^{\ot \ell}$ and the Lax operators are given simply by the $R$-matrix \eqref{Ru}, a nearest neighbour spin chain Hamiltonian may be extracted from the transfer matrix by taking the logarithmic derivative at a particular value of $v$. Indeed, set $c_i = 1$ for $1 \leq i \leq \ell$ and define an adjusted transfer matrix by 
\[
t(v) := \bigg(\frac{1-v}{q-q^{-1}}\bigg)^\ell \tau(v),
\]
with the property that at $v=1$ it becomes the shift operator,
\[
t(1) = \tr_a P_{a1} \, P_{a2} \cdots P_{a \ell} = P_{\ell-1,\ell} \,P_{\ell-2,\ell-1} \cdots P_{12}.
\]
A nearest-neighbour interaction Hamiltonian is then 
\[
H := \frac{d}{dv} \ln t(v) \Big|_{v=1} = \big( t(1)^{-1} \big)\, t'(1) = \sum_{i=1}^{\ell-1} h_{i,i+1} + h_{\ell,1},
\]
where the interaction between adjacent sites $h \in \End(\C^n \ot \C^n)$ is given by
\[
h := I -\frac{P\,R_{q}}{q-q^{-1}}  + \frac{P\,Q_{q}}{q^{2\kappa}-1} .
\]

%%%%%%%%%%%%%%%%%%%%%%%%%%%%%%%%%%%%%%%%%%%%%%%%%%%%%%%%%%%%%%%%%%%%
%
%%%%%%%%%%%%%%%%%%%%%%%%%%%%%%%%%%%%%%%%%%%%%%%%%%%%%%%%%%%%%%%%%%%%

\subsection{Trace formula for Bethe vectors}

Introduce matrices $f_\theta \in \End(\C^{2n})$ by $f_{+} = q^{-1}e_{n+2,n}-q\,e_{n+1,n-1}$ and $f_{-} = - q\, e_{n+1,n}$. Then define a transposition $\om$ on $\End(\C^{2n})$ by $\om : e_{ij} \mapsto \theta_{ij} q^{\nu_i - \nu_j} e_{\bar\jmath \bar\imath}$ where $\bar\imath = 2n-i+1$ and $\bar\jmath = 2n-j+1$. 
The Theorem below is our second main result.

\begin{thrm} \label{T:tf}
The level-0 Bethe vector can be written as
\ali{ \label{TF}
\Phi^{(0)}_\theta(\bm u^{(0...n-1)}) &= \tr_{\ol{W}} \Bigg[\Bigg( \prod_{k=1}^{n-1} \prod_{i=1}^{m_k} \prod_{j=1}^{m_0} R_{q^{-1},a_i^k a^{0}_j}(u^{(k)}_i,q^{2\theta} u^{(0)}_j) \Bigg) \el & \qq\qq \times \Bigg( \prod_{i=1}^{m_0} T^\om_{a^0_i} (u^{(0)}_i)  \Bigg)
\Bigg( \prod_{k=1}^{n-1} \prod_{i=1}^{m_k} \prod_{j=1}^{m_{0}} R_{a_i^k a^{0}_j}(u^{(k)}_i,q^{-2 \kappa} u^{(0)}_j) \Bigg)
 \el
& \qq\qq\times 
\Bigg( \prod_{k=1}^{n-1} \prod_{i=1}^{m_k}  T_{a_i^k}(u^{(k)}_i) \Bigg)  \Bigg( \prod_{k=2}^{n-1} \prod_{l=1}^{k-1} \prod_{i=1}^{m_k} \prod_{j=m_l}^{1} \!R_{a^k_i a^l_j}(u^{(k)}_i, u^{(l)}_j) \Bigg) 
 \el & \qq\qq  \times (f_{\theta})^{\ot m_0} \ot (e_{21})^{\ot m_1} \ot \cdots \ot (e_{n,n-1})^{\ot m_{n-1}} \Bigg]\cdot \eta \,,
}
where the trace is taken over the space $\ol{W} = W_{\bm a^0} \ot \cdots \ot W_{\bm a^{n-1}}  \cong (\C^{2n})^{\ot (m_0 + \ldots + m_{n-1})}$.
\end{thrm}

\begin{proof} 
Recall the trace formula for the Bethe vectors of a $U_q(\mfgl_n)$-symmetric spin chain given in Section 5.2 of \cite{BeRa08}. 
This result implies the following formula for the level-$1$ Bethe vector: 
\aln{
\Phi^{(1)}_\theta(\bm u^{(1...n-1)};\bm u^{(0)} ) &= \tr_{\ol{W}^{(1)}} \Bigg[
\Bigg( \prod_{k=1}^{n-1} \prod_{i=1}^{m_k} A^{(1)}_{a_i^k \bm a^{\tl 0,0}}(u^{(k)}_i;\bm u^{(0)}) \Bigg) \Bigg( \prod_{k=2}^{n-1} \prod_{l=1}^{k-1} \prod_{i=1}^{m_k} \prod_{j=m_l}^{1} \!R^{(1,1)}_{a^k_i a^l_j}(u^{(k)}_i, u^{(l)}_j) \Bigg) \el & \qq\qq\qu \times  (e^{(1)}_{21})^{\ot m_1} \ot \cdots \ot (e^{(1)}_{n,n-1})^{\ot m_{n-1}} \Bigg]\cdot \eta^{(1)}_{-\theta},
}
where the trace is taken over the space $\ol{W}^{(1)} = W^{(1)}_{a^1_1} \ot \cdots \ot W^{(1)}_{a^{n-1}_{m_{n-1}}} \cong (\C^{n})^{\ot (m_1 + \ldots + m_{n-1})}$.
From \eqref{A1}, this is equal to
\aln{
\Phi^{(1)}_\theta(\bm u^{(1...n-1)};\bm u^{(0)} ) &= \tr_{\ol{W}^{(1)}} \Bigg[
\Bigg( \prod_{k=1}^{n-1} \prod_{i=1}^{m_k} \prod_{j=m_0}^{1} K^{(1,1)}_{a_i^k \tl a^{0}_j}(u^{(k)}_i,q^{2\theta} u^{(0)}_j) \Bigg)
\Bigg( \prod_{k=1}^{n-1} \prod_{i=1}^{m_k} \prod_{j=1}^{m_0} K^{(1,1)}_{a_i^k a^{0}_j}(u^{(k)}_i, u^{(0)}_j) \Bigg)
\el 
& \qq\qq\; \times \Bigg( \prod_{k=1}^{n-1} \prod_{i=1}^{m_k}  A^{(1)}_{a_i^k}(u^{(k)}_i) \Bigg) \Bigg( \prod_{k=2}^{n-1} \prod_{l=1}^{k-1} \prod_{i=1}^{m_k} \prod_{j=m_l}^{1} \!R^{(1,1)}_{a^k_i a^l_j}(u^{(k)}_i, u^{(l)}_j) \Bigg) 
\el 
& \qq\qq\; \times  (e^{(1)}_{21})^{\ot m_1} \ot \cdots \ot (e^{(1)}_{n,n-1})^{\ot m_{n-1}} \Bigg]\cdot \eta^{(1)}_{-\theta} .
}
We now introduce the level-0 creation operators in order to arrive at an expression for the level-0 Bethe vector, as given in Definition~\ref{D:BV},
\aln{
\Phi^{(0)}_\theta(\bm u^{(0...n-1)}) &= \tr_{\ol{W}^{(1)}} \Bigg[  \Bigg( \prod_{i=1}^{m_0} \be_{\ta^0_i a^0_i} (u^{(0)}_i)  \Bigg)
\Bigg( \prod_{k=1}^{n-1} \prod_{i=1}^{m_k} \prod_{j=m_0}^{1} K^{(1,1)}_{a_i^k \tl a^{0}_j}(u^{(k)}_i,q^{2\theta} u^{(0)}_j) \Bigg)
 \el
 & \qq\qq\; \times \Bigg( \prod_{k=1}^{n-1} \prod_{i=1}^{m_k} \prod_{j=1}^{m_{0}} K^{(1,1)}_{a_i^k  a^{0}_j}(u^{(k)}_i, u^{(0)}_j) \Bigg) \Bigg( \prod_{k=1}^{n-1} \prod_{i=1}^{m_k}  A^{(1)}_{a_i^k}(u^{(k)}_i) \Bigg) 
 \el
& \qq\qq\; \times \Bigg( \prod_{k=2}^{n-1} \prod_{l=1}^{k-1} \prod_{i=1}^{m_k} \prod_{j=m_l}^{1} \!R^{(1,1)}_{a^k_i a^l_j}(u^{(k)}_i, u^{(l)}_j) \Bigg)   (e^{(1)}_{21})^{\ot m_1} \ot \cdots \ot (e^{(1)}_{n,n-1})^{\ot m_{n-1}} \Bigg]\cdot \eta^{(1)}_{-\theta} .
}
The next step is to rewrite the above expression in terms of the matrix $B(u)$, cf. \eqref{Lu:new}, rather than the creation operator $\beta(u)$. 
Consider the following in expression, in which a matrix operator $X$ acts non-trivially on the space $\ta^0$ and trivially on the space $a^0$, and vice versa for the matrix operator $Y$, and both operators act non-trivially on any number of other spaces,
\[
\beta_{\ta^0 a^0}(u)\,  X_{\ta^0} \, Y_{a^0} \cdot (e_k^{(1)})_{\ta^0} \ot (e_l^{(1)})_{a^0} = \sum_{ij}  (e^{(1)*}_j)_{\ta^0} \ot (e^{(1)*}_i)_{a^0} \ot  q^{-i} b_{\bi j}(u) \cdot X_{\ta^0} \, Y_{a^0} \cdot (e_k^{(1)})_{\ta^0} \ot (e_l^{(1)})_{a^0} .
\]
Contracting matrices gives
\aln{
\sum_{ij} q^{-i}\,b_{\bi j}(u)\,[X]_{jk} [Y]_{il} &=\sum_{ij} q^{-\bar\jmath} \big[ B^\om(u)\big]_{\bj i} [X]_{jk} [Y]_{il}
\\
&=\sum_{ij} q^{-\bar k} \big[ B^\om(u)\big]_{\bj i} [X^\om]_{\bar k \bj} [Y]_{il}
\\
&=\sum_{ij} q^{-\bar k} [X^{ \omega}]_{\bar k \bj} \big[ B^{ \omega}(u)\big]_{\bj i}  [Y]_{il}
\\
&= q^{-\bar k} [X^{ \omega} \, B^{ \omega}(u) \, Y]_{\bar k l} = q^{-\bar k} \tr \big[X^{ \omega} \, B^{ \omega}(u) \, Y e^{(1)}_{l \bar k} \big] .
}
We have thus arrived at the identity
\[
\beta_{\ta^0 a^0}(u)\,  X_{\ta^0} \, Y_{a^0} \cdot (e_k^{(1)})_{\ta^0} \ot (e_l^{(1)})_{a^0} = q^{-\bar k} \tr \big[X^{ \omega} \, B^{ \omega}(u) \, Y e^{(1)}_{l \bar k} \big] .
\]
Now recall \eqref{xi}. Hence we need to consider the cases when $(k,l)=(1,1),(1,2),(2,1)$, or equivalently $(l,\bar k)=(1,n),(2,n),(1,n-1)$. Bearing this in mind we define matrices $f^{(1)}_\theta \in \End(\C^n)$ by $f^{(1)}_- = -q\,e^{(1)}_{1,n}$ and $f^{(1)}_+ = q^{-1} e^{(1)}_{2,n} - q\,e^{(1)}_{1,n-1}$.
This allows us to write the level-$0$ Bethe vector as follows:
\ali{
&\Phi^{(0)}_\theta(\bm u^{(0...n-1)}) 
\el
&\qq =  \tr_{\ol{W}} \Bigg[\Bigg( \prod_{k=1}^{n-1} \prod_{i=1}^{m_k} \prod_{j=m_0}^{1} K^{(1,1)}_{a_i^k a^{0}_j}(u^{(k)}_i,q^{2\theta} u^{(0)}_j) \Bigg)^{\omega_{\bm a^0}}  \Bigg( \prod_{i=1}^{m_0} B^\om_{a^0_i} (u^{(0)}_i)  \Bigg) \Bigg( \prod_{k=1}^{n-1} \prod_{i=1}^{m_k} \prod_{j=1}^{m_{0}} K^{(1,1)}_{a_i^k a^{0}_j}(u^{(k)}_i, u^{(0)}_j) \Bigg)
\el
& \qq \qq \qq \times 
\Bigg( \prod_{k=1}^{n-1} \prod_{i=1}^{m_k}  A^{(1)}_{a_i^k}(u^{(k)}_i) \Bigg) \Bigg( \prod_{k=2}^{n-1} \prod_{l=1}^{k-1} \prod_{i=1}^{m_k} \prod_{j=m_l}^{1} \!R^{(1,1)}_{a^k_i a^l_j}(u^{(k)}_i, u^{(l)}_j) \Bigg)
 \el & \qq \qq \qq \times (\theta q^{\theta-n}f^{(1)}_\theta)^{\ot m_0} \ot (e^{(1)}_{21})^{\ot m_1} \ot \cdots \ot (e^{(1)}_{n,n-1})^{\ot m_{n-1}} \Bigg]\cdot \eta
 \el
&\qq =  \tr_{\ol{W}} \Bigg[\Bigg( \prod_{k=1}^{n-1} \prod_{i=1}^{m_k} \prod_{j=1}^{m_0} R^{(1,1)}_{q^{-1},a_i^k a^{0}_j}(u^{(k)}_i,q^{2\theta} u^{(0)}_j) \Bigg) \Bigg( \prod_{i=1}^{m_0} \theta q^{\theta-n} B^\om_{a^0_i} (u^{(0)}_i)  \Bigg) \Bigg( \prod_{k=1}^{n-1} \prod_{i=1}^{m_k} \prod_{j=1}^{m_{0}} K^{(1,1)}_{a_i^k a^{0}_j}(u^{(k)}_i, u^{(0)}_j) \Bigg)
\el
& \qq \qq \qq \times 
\Bigg( \prod_{k=1}^{n-1} \prod_{i=1}^{m_k}  A^{(1)}_{a_i^k}(u^{(k)}_i) \Bigg) \Bigg( \prod_{k=2}^{n-1} \prod_{l=1}^{k-1} \prod_{i=1}^{m_k} \prod_{j=m_l}^{1} \!R^{(1,1)}_{a^k_i a^l_j}(u^{(k)}_i, u^{(l)}_j) \Bigg)
 \el & \qq \qq \qq \times (f^{(1)}_\theta)^{\ot m_0} \ot (e^{(1)}_{21})^{\ot m_1} \ot \cdots \ot (e^{(1)}_{n,n-1})^{\ot m_{n-1}} \Bigg]\cdot \eta. \nn\\[-2.3em] \label{TFX-red}
}

It remains to show that the form of the Bethe vector given in \eqref{TF} reduces to to above form, by considering the decomposition $\C^{2n} \cong \C^2 \ot \C^n$ as in \eqref{e=x*e}, and tracing out all the $\C^2$ spaces.
Making explicit this decomposition, the formula \eqref{TF} becomes
\ali{
\label{TFX}
\Phi^{(0)}_\theta(\bm u^{(0...n-1)}) & = \tr_{\ol{W}} \Bigg[\Bigg( \prod_{k=1}^{n-1} \prod_{i=1}^{m_k} \prod_{j=1}^{m_0} R_{q^{-1},a_i^k a^{0}_j}(u^{(k)}_i,q^{2\theta} u^{(0)}_j) \Bigg) 
 \el
& \qq\qq \times 
\Bigg( \prod_{i=1}^{m_0} T^\om_{a^0_i} (u^{(0)}_i)  \Bigg) \Bigg( \prod_{k=1}^{n-1} \prod_{i=1}^{m_k} \prod_{j=1}^{m_{0}} R_{a_i^k a^{0}_j}(u^{(k)}_i,q^{-2 \kappa} u^{(0)}_j) \Bigg)  
\el
& \qq\qq \times  \Bigg( \prod_{k=1}^{n-1} \prod_{i=1}^{m_k}  T_{a_i^k}(u^{(k)}_i) \Bigg) \Bigg( \prod_{k=2}^{n-1} \prod_{l=1}^{k-1} \prod_{i=1}^{m_k} \prod_{j=m_l}^{1} \!R_{a^k_i a^l_j}(u^{(k)}_i, u^{(l)}_j) \Bigg) 
 \el & \qq\qq \times (x_{21} \ot f^{(1)}_{\theta})^{\ot m_0} \ot (x_{11} \ot e^{(1)}_{21})^{\ot m_1} \ot \cdots \ot (x_{11} \ot e^{(1)}_{n,n-1})^{\ot m_{n-1}} \Bigg]\cdot \eta.
}
Recall \eqref{R(u):new} and note that
\gan{
R(u,v) \;x_{11} \ot x_{11} = R^{(1,1)}(u,v)\, x_{11} \ot x_{11} ,
\\
R(u,v) \;x_{11} \ot x_{21} = U^{(1,1)}(u,v) \,  x_{21} \ot x_{11} + K^{(1,1)}(u,q^{2\kappa}v) \, x_{11} \ot x_{21} .
}
Next, recall \eqref{nu} and note that
\[
\nu_{n+j}-\nu_i = j - \theta' - (-n+i-1+\theta') = j-i+n-\theta
\]
for all $1 \leq i,j \leq n$.
This implies the following relationship between the transposition $\omega$ on $\End(\C^{2n})$ and its counterpart on $\End(\C^n)$:
\aln{
\big[ T^\om(u)\big]_{i,n+j} &=  \theta q^{i-j-n+\theta} t_{n-j+1, 2n-i+1}(u) =\theta q^{\theta-n} \big[ B^\om(u)\big]_{ij} .
}
Hence, the action of $T(u)$ and $T^\om(u)$ on $x_{11}$ and $x_{21}$ takes the form
\gan{
T(u) \, x_{11} =  A(u)\,x_{11}+ C(u)\,x_{21} ,
\\
T^\om(u) \, x_{21} = \theta q^{\theta-n} B^\om(u)\,x_{11} + A^\om(u)\,x_{21} .
}
The identities above imply that the numbers of $x_{11}$'s and $x_{21}$'s inside the trace in \eqref{TFX} are conserved individually under the action $R$-matrices.
Therefore, the only possibility for the partial trace over the $\C^2$ spaces to be nonzero is if the action of the $T_{a^0_i}^\om(u^{(0)}_i)$ maps $(x_{21})_{a^0_i}$ to $(x_{11})_{a^0_i}$. 
That is, each  $T_{a^0_i}^\om(u^{(0)}_i)$ acts as $ q^{\theta-n} \theta B_{a^0_i}^\om(u^{(0)}_i)$.
Since each $(x_{21})_{a^0_i}$ must be acted on by $T_{a^0_i}^\om(u^{(0)}_i)$, the $R$-matrices to the right of the $T_{a^0_i}^\om(u^{(0)}_i)$'s must not permute the $(x_{21})_{a^0_j}$ with the $(x_{11})_{a^k_i}$ for $k \geq 1$. 
That is, the $R_{a_i^k a^{0}_j}(u^{(k)}_i,q^{-2 \kappa} u^{(0)}_j)$ each act as $K^{(1,1)}_{a_i^k a^{0}_j}(u^{(k)}_i,u^{(0)}_j)$ in order for the trace to be non-zero.
Finally, all other $R$-matrices in \eqref{TFX} act on suitable pairs $x_{11} \ot x_{11}$ only, and may simply be replaced with $R^{(1,1)}(u,v)$, for appropriate $u,v$.
This proves that taking the partial trace over the $\C^2$ spaces in \eqref{TFX} we arrive at \eqref{TFX-red}, as required.
\end{proof}

%%% Bethe Vector Examples

Below we provide two examples of level-0 Bethe vectors obtained using \eqref{TFX}. We will assume that $m_i =0$ for $0\le i \le n-1$ if not stated otherwise. We also set $p_k := \del_{k,n-1}\,q^{-2}\,(q-q^{-1})$.

\begin{exam}[Symplectic case]
For $n\ge 1$ and $m_0=m\ge 1$ we have
\[
\Phi^{(0)}_{-}(u^{(0)}_1,\dots,u^{(0)}_m) = q^{-m}t_{n,n+1}(u^{(0)}_1) \cdots t_{n,n+1}(u^{(0)}_m) \cdot \eta .
\]
For $n\ge 2$ and $m_0=m_k=1$ with $1\le k\le n-1$ we have
\[
\Phi^{(0)}_{-}(u^{(0)}_1,u^{(k)}_1) = q^{-1}\Bigg[\, t_{n,n+1}(u^{(0)}_1)\, t_{k,k+1}(u^{(k)}_1) + \frac{p_k\,u^{(0)}_1}{u^{(k)}_1-u^{(0)}_1} \Big(t_{n-1,n+1}(u^{(0)}_1)+t_{n,n+2}(u^{(0)}_1)\Big)\,t_{nn}(u^{(k)}_1)\,\Bigg] \cdot \eta .
\]
For $n\ge 2$ and $m_0=2$, $m_k=1$ with $1\le k\le n-1$ we have
\aln{
& \Phi^{(0)}_{-}(u^{(0)}_1,u^{(0)}_2,u^{(k)}_1) = q^{-2}\Bigg[\, t_{n,n+1}(u^{(0)}_1)\,t_{n,n+1}(u^{(0)}_2)\, t_{k,k+1}(u^{(k)}_1) \\ 
& \qu + p_k \Bigg( \frac{u^{(0)}_2}{u^{(k)}_1-u^{(0)}_2}\, t_{n,n+1}(u^{(0)}_1) \Bigg(t_{n-1,n+1}(u^{(0)}_2) + \frac{q u^{(k)}_1-q^{-1}u^{(0)}_1}{u^{(k)}_1-u^{(0)}_1}\, t_{n,n+2}(u^{(0)}_2) \Bigg) \\
& \qq\qq + \frac{u^{(0)}_1}{u^{(k)}_1-u^{(0)}_1} \Bigg(\frac{q u^{(k)}_1-q^{-1}u^{(0)}_2}{u^{(k)}_1-u^{(0)}_2}\,t_{n-1,n+1}(u^{(0)}_1) + \frac{q^2 u^{(k)}_1-q^{-2}u^{(0)}_2}{u^{(k)}_1-u^{(0)}_2}\,t_{n,n+2}(u^{(0)}_1) \Bigg) t_{n,n+1}(u^{(0)}_2) \Bigg)\\ & \hspace{13cm} \times t_{k+1,k+1}(u^{(k)}_1) \,\Bigg] \cdot \eta .
}
\end{exam}

\begin{exam}[Orthogonal case]
For $n\ge 1$ and $m_0=m\ge 1$ we have
\[
\Phi^{(0)}_{+}(u^{(0)}_1,\dots,u^{(0)}_m) = \prod_{k=1}^m \Big(q^{-2}t_{n-1,n+1}(u^{(0)}_k) -t_{n,n+2}(u^{(0)}_k)\Big) \cdot \eta .
\]
For $n\ge 2$ and $m_0=m_k=1$ with $1\le k\le n-1$ we have
\aln{
\Phi^{(0)}_{+}(u^{(0)}_1,u^{(k)}_1) &= \Bigg[\,\frac{q^{\del_{k,n-1}} u^{(k)}_1 -q^{-\del_{k,n-1}} u^{(0)}_1}{q^2(u^{(k)}_1-u^{(0)}_1)} \left(q^{-2} t_{n-1,n+1}(u^{(0)}_1)-t_{n,n+2}(u^{(0)}_1)\right)\, t_{k,k+1}(u^{(k)}_1) \\ 
& \hspace{2.8cm} + \frac{p_{k+1}\,u^{(0)}_1}{u^{(k)}_1-u^{(0)}_1} \left(q^{-2} t_{n-2,n+1}(u^{(0)}_1) - q\,t_{n,n+3}(u^{(0)}_1) \right) t_{k+1,k+1}(u^{(k)}_1) \,\Bigg] \cdot \eta .
}
For $n\ge 4$ and $m_0=2$, $m_k=1$ with $1\le k\le n-3$ we have
\aln{
\Phi^{(0)}_{+}(u^{(0)}_1,u^{(0)}_2,u^{(k)}_1) &= \left(q^{-2} t_{n-1,n+1}(u^{(0)}_1)-t_{n,n+2}(u^{(0)}_1)\right) \\ & \qu\times \left(q^{-2} t_{n-1,n+1}(u^{(0)}_2)-t_{n,n+2}(u^{(0)}_2)\right) \, t_{k,k+1}(u^{(k)}_1)\cdot \eta .
}
The $k=n-2$ and $k=n-1$ cases have long tails, hence we have not written them out explicitly. 
\end{exam}

%%%%%%%%%%%%%%%%%%%%%%%%%%%%%%%%%%%%%%%%%%%%%%%%%%%%%%%%%%%%%%%%%%
% Appendix
%%%%%%%%%%%%%%%%%%%%%%%%%%%%%%%%%%%%%%%%%%%%%%%%%%%%%%%%%%%%%%%%%%

\appendix

%%%%%%%%%%%%%%%%%%%%%%%%%%%%%%%%%%%%%%%%%%%%%%%%%%%%%%%%%%%%%%%%%%%%

\nc{\qlim}{\underset{\hbar\to 0}{\longrightarrow}}

\section{The semi-classical limit}

In order to retrieve the results of \cite{Res85} and \cite{DVK87}, we investigate the semi-classical $q \rightarrow 1$ limit, or equivalently the $\hbar \rightarrow 0$ limit. 
The limit must be taken in a particular way, as the spectral parameters have a hidden $q$ dependence. 
Setting $u = e^{2x\hbar}$, $v = e^{2y\hbar}$ and $q = e^{\hbar}$ and expanding in powers of $\hbar$, we recover the Zamolodchikov $R$-matrix \cite{ZaZa78,KuSk82},
\[
R_q \qlim I,  \qq  Q_q \qlim Q := \sum_{i,j=1}^{2n} e_{ij} \ot e_{\bi \bj}, \qq
R(u,v) \qlim R(x-y) := I - \frac{P}{x - y}   - \frac{Q}{\ka - (x - y)} \,.
\]
The reduced $R$-matrices become the Yang $R$-matrices,
\aln{
R^{(k,l)}(u,v) \qlim R^{(k,l)}(x-y) := I^{(k,l)} - \frac{P^{(k,l)}}{x - y} .
}
The eigenvalues given in Theorem \ref{T:gln} in the limit become
\aln{
\La^{(1)}(y;\bm x^{(1\dots n-1)}; \bm x^{(0)}) & = \prod_{i=1}^{m_1} \frac{y-x^{(1)}_i+1}{y-x^{(1)}_i} \prod_{i=1}^\ell\la^{(i)}_{1}(y) \el
& \qu + \sum_{k=2}^{n-2} \prod_{i=1}^{m_{k-1}} \frac{y-x^{(k-1)}_i-1}{y-x^{(k-1)}_i} \prod_{i=1}^{m_k} \frac{y-x^{(k)}_i+1}{y-x^{(k)}_i} \prod_{i=1}^\ell\la^{(i)}_{k}(y)  \el 
& \qu + \prod_{i=1}^{m_{n-2}} \frac{y-x^{(n-2)}_i-1}{y-x^{(n-2)}_i} \prod_{i=1}^{m_{n-1}} \frac{y-x^{(n-1)}_i+1}{y-x^{(n-1)}_i} \prod_{i=1}^{m_0} \frac{y-x^{(0)}_i+\theta'}{y-x^{(0)}_i} \prod_{i=1}^\ell\la^{(i)}_{n-1}(y) \el
& \qu +  \prod_{i=1}^{m_{n-1}} \frac{y-x^{(n-1)}_i-1}{y-x^{(n-1)}_i} \prod_{i=1}^{m_0} \frac{y-x^{(0)}_i+2-\theta'}{y-x^{(0)}_i} \prod_{i=1}^\ell\la^{(i)}_{n}(y) \\[-.75em] 
\intertext{and}
\wt{\La}^{(1)}(y;\bm x^{(1\dots n-1)};\bm x^{(0)}) & = \prod_{i=1}^{m_1} \frac{y-x^{(1)}_i - \ka }{y-x^{(1)}_i-\ka+1} \prod_{i=1}^\ell\la^{(i)}_{2n}(y) \el
& \qu + \sum_{k=2}^{n-2} \prod_{i=1}^{m_{k-1}} \frac{y-x^{(k-1)}_i-\ka+k}{y-x^{(k-1)}_i-\ka+k-1} \prod_{i=1}^{m_k} \frac{y-x^{(k)}_i-\ka+k-1}{y-x^{(k)}_i-\ka+k} \prod_{i=1}^\ell\la^{(i)}_{2n-k+1}(y)  
\el 
& \qu + \prod_{i=1}^{m_{n-2}} \frac{y-x^{(n-2)}_i+\theta-1}{y-x^{(n-2)}_i+\theta-2} \prod_{i=1}^{m_{n-1}} \frac{y-x^{(n-1)}_i+\theta-1}{y-x^{(n-1)}_i+\theta-2}  \prod_{i=1}^{m_0} \frac{y-x^{(0)}_i-\theta'}{y-x^{(0)}_i} \prod_{i=1}^\ell\la^{(i)}_{n+2}(y)  \el 
& \qu + \prod_{i=1}^{m_{n-1}} \frac{y-x^{(n-1)}_i+\theta}{y-x^{(n-1)}_i+\theta-1} \prod_{i=1}^{m_0} \frac{y-x^{(0)}_i-2 +\theta'}{y-x^{(0)}_i} \prod_{i=1}^\ell\la^{(i)}_{n+1}(y)
}
where the rational weights are given by 
\ali{
\la^{(i)}_j(v) \qlim \la^{(i)}_j(y) &:= \begin{cases} 
\dfrac{y- c_i - s_i}{ y - c_i }  &\text{if} \qu j=1, \\[.75em]
1 &\text{if}\qu 1< j < 2n, \\[.25em] 
\dfrac{y-b_i-\ka+1}{y-b_i-\ka+1-s_i} & \text{if}\qu j= 2n  \end{cases}  \label{la-symm-rat}   
\intertext{in the symmetric case, i.e.\ when $\mfg_{2n}=\mfso_{2n}$, and by}
\la^{(i)}_j(v) \qlim \la^{(i)}_j(y) & := \begin{cases} 
\dfrac{y-b_i-1}{y-b_i}  &\text{if}\qu 1\le j \le s_i, \\[.75em]
1 &\text{if}\qu s_i< j < 2n-s_i+1, \\[.25em] 
\dfrac{y-b_i-\ka+s_i}{y-b_i-\ka+s_i-1} & \text{if}\qu 2n-s_i+1\le j\le 2n  \end{cases} \label{la-skew-rat} 
}
in the skewsymmetric case, i.e.\ when $\mfg_{2n}=\mfsp_{2n}$; here $b_i = \frac{1}{2\hbar}\log c_i\in \C$ are the inhomogeneities.

The Bethe equations may be obtained in the same way. Denoting $x^{(n)}_j := x^{(0)}_j$ and $m_n := m_0$ their explicit form is, in the symplectic case,
\aln{
\prod_{i=1}^\ell\frac{\la^{(i)}_{1}(x_j^{(1)})}{\la^{(i)}_{2}(x_j^{(1)})} &= \prod_{\substack{i=1\\ i\ne j}}^{m_1} \frac{x_j^{(1)}-x^{(1)}_i-1}{x_j^{(1)}-x^{(1)}_i+1} \prod_{i=1}^{m_{2}} \frac{x_j^{(1)}-x^{(2)}_i+1}{x_j^{(1)}-x^{(2)}_i} \,, 
\\
\prod_{i=1}^\ell \frac{\la^{(i)}_{k}(x^{(k)}_j)}{\la^{(i)}_{k+1}(x^{(k)}_j)} &= \prod_{i=1}^{m_{k-1}} \frac{x^{(k)}_j-x^{(k-1)}_i}{x^{(k)}_j-x^{(k-1)}_i-1}  \prod_{\substack{i=1\\ i\ne j}}^{m_{k}} \frac{x^{(k)}_j-x^{(k)}_i-1}{x^{(k)}_j-x^{(k)}_i+1} \prod_{i=1}^{m_{k+1}} \frac{x^{(k)}_j-x^{(k+1)}_i+1}{x^{(k)}_j-x^{(k+1)}_i} \,,
\\
\prod_{i=1}^\ell \frac{\la^{(i)}_{n-1}(x^{(n-1)}_j)}{\la^{(i)}_{n}(x^{(n-1)}_j)} &= \prod_{i=1}^{m_{n-2}} \frac{x^{(n-1)}_j-x^{(n-2)}_i}{x^{(n-1)}_j-x^{(n-2)}_i-1} \prod_{\substack{i=1\\ i\ne j}}^{m_{n-1}} \frac{ x^{(n-1)}_j-x^{(n-1)}_i-1}{x^{(n-1)}_j-x^{(n-1)}_i+1} \prod_{i=1}^{m_n} \frac{x^{(n-1)}_j-x^{(n)}_i+2}{x^{(n-1)}_j-x^{(n)}_i} \,,
\\
\prod_{i=1}^\ell \frac{\la^{(i)}_{n}(x^{(n)}_j)}{\la^{(i)}_{n+1}(x^{(n)}_j)} &=  \prod_{i=1}^{m_{n-1}} \frac{x^{(n)}_j-x^{(n-1)}_i}{x^{(n)}_j-x^{(n-1)}_i-2} \prod_{\substack{i=1\\i\ne j}}^{m_n} \frac{x^{(n)}_j-x^{(n)}_i-2}{x^{(n)}_j-x^{(n)}_i+2} 
\intertext{for $2\le k \le n-2$ and all allowed $j$, and weights given by \eqref{la-skew-rat}.
In the orthogonal case, the Bethe equations for $k = 1, \ldots , n-3$ are identical to the symplectic case; for $k = n-2, n-1, n$ they are, respectively,
}
\prod_{i=1}^\ell \frac{\la^{(i)}_{n-2}(x^{(n-2)}_j)}{\la^{(i)}_{n-1}(x^{(n-2)}_j)} &= \prod_{i=1}^{m_{n-3}} \frac{x^{(n-2)}_j-x^{(n-3)}_i}{ x^{(n-2)}_j-x^{(n-3)}_i-1} \prod_{\substack{i=1\\i\ne j}}^{m_{n-2}} \frac{x^{(n-2)}_j-x^{(n-2)}_i-1}{x^{(n-2)}_j-x^{(n-2)}_i+1} \el & \qu \times \prod_{i=1}^{m_{n-1}} \frac{x^{(n-2)}_j-x^{(n-1)}_i+1}{x^{(n-2)}_j-x^{(n-1)}_i} \prod_{i=1}^{m_n} \frac{x^{(n-2)}_j-x^{(n)}_i+1}{x^{(n-2)}_j-x^{(n)}_i} \,,
\\[0.5em]
\prod_{i=1}^\ell \frac{\la^{(i)}_{n-1}(x^{(n-1)}_j)}{\la^{(i)}_{n}(x^{(n-1)}_j) } &= \prod_{i=1}^{m_{n-2}} \frac{x^{(n-1)}_j-x^{(n-2)}_i}{x^{(n-1)}_j-x^{(n-2)}_i-1} \prod_{\substack{i=1\\i\ne j}}^{m_{n-1}} \frac{ x^{(n-1)}_j-x^{(n-1)}_i-1}{x^{(n-1)}_j-x^{(n-1)}_i+1}  \,, 
\\[0.5em]
 \prod_{i=1}^\ell \frac{\la^{(i)}_{n-1}(x^{(n)}_j)}{\la^{(i)}_{n+1}(x^{(n)}_j)} &=  \prod_{i=1}^{m_{n-2}} \frac{x^{(n)}_j-x^{(n-2)}_i}{x^{(n)}_j-x^{(n-2)}_i-1}  \prod_{\substack{i=1\\i\ne j}}^{m_n} \frac{x^{(n)}_j-x^{(n)}_i-1}{x^{(n)}_j-x^{(n)}_i+1} 
}
for all allowed $j$, and weights given by \eqref{la-symm-rat}. Substituting $x^{(k)}_j \to w^{(k)}_j - \wt d_k$ with $\wt d_k$ and assuming restrictions on $n$ as in Remark \ref{R:Bethe-Dynkin}, the Bethe equations for both symplectic and orthogonal cases take the form
\[
\prod_{i=1}^\ell \frac{\la^{(i)}_k(w^{(k)}_j - \wt d_k)}{\la^{(i)}_{k+1}(w^{(k)}_j - \wt d_k)} = - \prod_{l \in I} \prod_{i \in 1}^{m_l} \frac{w^{(k)}_j-w^{(l)}_i - \frac12d_k a_{kl}}{w^{(k)}_j-w^{(l)}_i + \frac12d_k a_{kl}} 
\]
for $1\le k \le n$ and all allowed $j$.

Finally, the trace formula for level-0 Bethe vector \eqref{TFX} takes the form
\aln{
\Phi^{(0)}(\bm x^{(0...n-1)}) & = \tr_{\ol{W}} \Bigg[\Bigg( \prod_{k=1}^{n-1} \prod_{i=1}^{m_k} \prod_{j=1}^{m_0} R_{a_i^k a^{0}_j}(x^{(k)}_i-x^{(0)}_j-\theta) \Bigg) \\
& \qq\qq \times \Bigg( \prod_{i=1}^{m_0} T^{t}_{a^0_i} (x^{(0)}_i)  \Bigg)
\Bigg( \prod_{k=1}^{n-1} \prod_{i=1}^{m_k} \prod_{j=1}^{m_{0}} R_{a_i^k a^{0}_j}(x^{(k)}_i - x^{(0)}_j + \kappa ) \Bigg)
 \el
& \qq\qq \times 
\Bigg( \prod_{k=1}^{n-1} \prod_{i=1}^{m_k}  T_{a_i^k}(x^{(k)}_i) \Bigg)  \Bigg( \prod_{k=2}^{n-1} \prod_{l=1}^{k-1} \prod_{i=1}^{m_k} \prod_{j=m_l}^{1} \!R_{a^k_i a^l_j}(x^{(k)}_i - x^{(l)}_j) \Bigg) 
 \el & \qq\qq \times  \, (f_{\theta})^{\ot m_0} \ot (e_{21})^{\ot m_1} \ot \cdots \ot (e_{n,n-1})^{\ot m_{n-1}} \Bigg]\cdot \eta \,,
}
where the trace is taken over the space $\ol{W} = W_{a^0_1} \ot \cdots \ot W_{a^{n-1}_{m_{n-1}}} \cong (\C^{2n})^{\ot (m_0 + \ldots + m_{n-1}})$ and $f_\theta \in \End(\C^{2n})$ is defined for orthogonal and symplectic cases respectively by $f_1 = e_{n+1,n-1}-e_{n+2,n}$ and $f_{-1} = e_{n+1,n}$, and $T_{a}(x)$'s are defined via the rational fusion procedure analogous to that in Section \ref{sec:reps} (see \cite{IMO12} and Section 3.1 in \cite{GeRe20}) and $t$ is the transposition defined by $t: e_{ij} \mapsto e_{\bar\jmath\hspace{0.3mm} \bar\imath}$.

%%%%%%%%%%%%%%%%%%%%%%%%%%%%%%%%%%%%%%%%%%%%%%%%%%%%%%%%%%%%%%%%%%
% Bibliography
%%%%%%%%%%%%%%%%%%%%%%%%%%%%%%%%%%%%%%%%%%%%%%%%%%%%%%%%%%%%%%%%%%

\end{document}